\documentclass{article}

    \usepackage[final]{neurips_2022}

\usepackage[utf8]{inputenc} 
\usepackage[T1]{fontenc}    
\usepackage{hyperref}       
\usepackage{url}            
\usepackage{booktabs}       
\usepackage{amsfonts}       
\usepackage{nicefrac}       
\usepackage{microtype}      
\usepackage{xcolor}         

\usepackage{mathtools}       
\usepackage{amsmath}
\usepackage{amsthm}
\usepackage[linesnumbered,ruled,lined,noend]{algorithm2e} 
\usepackage{accents}
\usepackage{amsmath,amssymb,amsthm}

\usepackage{enumitem}
\setlist[itemize]{topsep=0pt,}

\usepackage{dsfont} 
\usepackage{xcolor} 
\usepackage[capitalise]{cleveref}
\usepackage{esvect}

\newcommand{\finishproof}[1]{{ This completes the proof of #1.}}

\def\##1\#{\begin{align}#1\end{align}}
\def\$#1\${\begin{align*}#1\end{align*}}
\def\given{\,|\,}

\def\tr{\mathop{\text{tr}}\kern.2ex}
\def \barF {{ \bar F }}
\def \cond {{\,|\,}}
\def \gam {{ \gamma }}
\def \indi {{ \mathds{1} }}

\def\R{{\mathbb R}}

\def\E{{\mathbb E}}

\def\cF{{\mathcal{F}}}

\def\cW{{\mathcal{W}}}

\def\cZ{{\mathcal{Z}}}

\def\cP{{\mathcal{P}}}

\def \sfv {{\mathsf{v}}}

\def \sfC {{ \mathsf{C} }}
\def \sfA {{ \mathsf{A} }}
\def \sfF {{ \mathsf{F} }}
\def \sfG {{ \mathsf{G} }}

\def\rA{{\mathrm{A}}}
\def\rB{{\mathrm{B}}}
\def\rC{{\mathrm{C}}}
\def\rD{{\mathrm{D}}}

\newtheorem{remark}{Remark}

\newtheorem{lemma}{Lemma}
\newtheorem{defn}{Definition}

\newlist{enumconditions}{enumerate}{1} 
\setlist[enumconditions]{label = \thelemma.\alph*}
\crefname{enumconditionsi}{Condition}{Conditions}

\newlist{enumlmresult}{enumerate}{1} 
\setlist[enumlmresult]{label = \thelemma.\arabic*}
\crefname{enumlmresulti}{Part}{Parts}

\newlist{enumthmresult}{enumerate}{1} 
\setlist[enumthmresult]{label = \thetheorem.\arabic*}
\crefname{enumthmresulti}{Part}{Parts}

\newlist{enumassumption}{enumerate}{1} 
\setlist[enumassumption]{label = \theAssumption.\arabic*}
\crefname{enumassumptioni}{Condition}{Conditions}

\newtheorem{theorem}{Theorem}
\renewcommand*{\thetheorem}{\arabic{theorem}}

\def \TV {{\scriptscriptstyle\mathrm{TV}}}
\def \block {\mathrm{block}}

\def \I {{\mathrm{{I}}}}
\def \II {{\mathrm{{II}}}}
\def \III {{\mathrm{{III}}}}

\def \Reg {{\mathrm{Reg}}}
\def \Envy {{\mathrm{Envy}}}

\def \cFtaum {{ \cF_{\tau - 1} }}

\def \betatp {{ \beta^{t+1}}}
\def \betatau {{ \beta^{\tau}}}

\def \betataui {{ \beta^\tau_i }}
\def \itau {{ i^\tau }}
\def \utaui {{ u^\tau_i }}

\def \ubartaui {{ \bar{u}^\tau_i }}
\def \vithetau {{ v_i(\theta^\tau) }}

\def \sumtau {{\sum_{\tau=1}^{t}}}
\def \sumk {{\sum_{k=1}^{K}}}
\def \sumtauinIk {{\sum_{\tau \in I_k}}}

\def \sumtautmi {{\sum_{\tau = 1 }^{t-\iota}}} 
\def \sumtaui {{ \sum_{\tau = 1}^\iota  }} 

\def \sumiton {{ \sum_{i=1}^n }}
\def \ztau {{z_\tau}}
\def \zt {{z_t}}
\def \zts {{ \{z_\tau\}_{\tau =1}^t }}
\def \ztaupi {{ z_{\tau + \iota} }}

\def \Ptaugiventaum {{ P^{\tau}(\cdot \,|\, z_{1:\tau-1} ) }}
\def \Ptaugiventaumdz {{ P^{\tau}(\diff z \,|\, z_{1:\tau-1} ) }}
\def \Ptaupigiventau {{ P^{\tau+ \iota}(\cdot \,|\, z_{1:\tau} ) }}

\def \Ptaupigiventaum {{ P^{\tau+ \iota}(\cdot \,|\, z_{1:\tau-1}) }}
\def \Ptaupigiventaumdz {{ P^{\tau+ \iota}(\diff z \,|\, z_{1:\tau-1}) }}

\def \tauk {{\tau_k}}
\def \taukp {{\tau_{k+1}}}

\def \thetau {{ \theta^\tau }}

\def \epsiota{{ \epsilon_t (\iota) }}

\def \sig {{ \sigma }}

\def \wtau {{w_\tau}}
\def \wtaupone {{w_{\tau + 1 }}}

\def \wst {{{w^*_\Pi}}}

\def \wtaupi {{ w_{\tau + \iota}}} 

\def \gtau {{ g_\tau }}

\def \gbartau {{\bar g_\tau}}
\def \gbartaum {{\bar g_{\tau - 1}}}

\def \vinfty {{ |v|_\infty }}
\def \vinftysq {{ |v|_\infty^2 }}

\def \betast {{\beta^*}}

\def \inv {^{-1}}
\def \sq {^{2}}

\def \st {^{*}}

\def \dom {{ \operatorname*{Dom}\,}}

\newcommand{\defeq}{\vcentcolon=}

\newcommand*\diff{\mathop{}\!\mathrm{d}}
\def \d {{\diff}}

\newcommand*\mix{\mathop{}\!\mathrm{mix}}

\DeclareMathOperator*{\argmax}{arg\,max}
\DeclareMathOperator*{\argmin}{arg\,min}

\hypersetup{
citebordercolor=blue
}
\def \NIPS {}
\newif \ifnips

\title{Nonstationary Dual Averaging and Online Fair Allocation}

\author{%
 Luofeng Liao, Yuan Gao, Christian Kroer \\ IEOR, Columbia University 
 \\ \texttt{\{ll3530,yg254,ck294\}@columbia.edu}
}

\begin{document}

\maketitle

\begin{abstract}
    \ifdefined \ARXIV
        \textbf{Abstract. }
    
    \fi
    
    We consider the problem of fairly allocating sequentially arriving items to a set of individuals. 
    For this problem, the recently-introduced PACE algorithm leverages the dual averaging algorithm to approximate competitive equilibria and thus generate online fair allocations. PACE is simple, distributed, and parameter-free, making it appealing for practical use in large-scale systems. However, current performance guarantees for PACE require i.i.d.\ item arrivals. Since real-world data is rarely i.i.d., or even stationary, we study the performance of PACE on nonstationary data.
    We start by developing new convergence results for the general dual averaging algorithm under three nonstationary input models: adversarially-corrupted stochastic input, ergodic input, and block-independent (including periodic) input. Our results show convergence of dual averaging up to errors caused by nonstationarity of the data, and recover the classical bounds when the input data is i.i.d.\ Using these results, we show that the PACE algorithm for online fair allocation simultaneously achieves ``best of many worlds'' guarantees against any of these nonstationary input models as well as against i.i.d. input. Finally, numerical experiments show strong empirical performance of PACE against nonstationary inputs. 
\end{abstract}

\section{Introduction}
In fair division, the goal is to allocate a set of items, typically assumed divisible, among a set of agents with heterogeneous preferences, while guaranteeing fairness and efficiency properties.
In this paper we are interested in how to fairly and efficiency allocate items that arrive \emph{online}: at every time step one item arrives, and we must irrevocably assign it to some agent. 
Recently, there has been a growing literature on such online fair allocation problems~\citep{azar2016allocate,balseiro2020best,gao2021online,bateni2021fair,sinclair2021sequential,banerjee2022online}. 
Real-world systems that can be captured by such settings include Internet advertising systems, job recommender systems, cloud computing platforms, and many more.
One of the key challenges in such problems is to balance the (often conflicting) goals of overall efficient resource utilization with fairness guarantees for the individual agents.


For this setting, \citet{gao2021online} shows that a simple mechanism called PACE (Pace According to Current Estimated utility) generates asymptotically fair and efficient allocations when the item arrivals are drawn in an i.i.d.\ manner. 
PACE gives each agent a per-time-step budget of faux currency, and the fair allocation is achieved by having agents participate in first-price auctions for each item, using the faux money.
By guaranteeing that each agent asymptotically spends their budget at the correct rate, the resulting allocations and prices converge to what is known as a \emph{competitive equilibrium from equal incomes} (CEEI), which guarantees both fairness and efficiency.
In PACE, each agent maintains a \emph{pacing multiplier} to control their spending over time, and the pacing multipliers are updated based on buyers' budgets and cumulative utilities. 
This is similar to how budget-management systems work in Internet ad auctions.
PACE is highly decentralized due to its auction-based allocation, it does not require dividing the item, and it is also completely parameter free. This makes it suitable for large-scale practical implementation.

Yet in many large-scale settings, such as the context of fair recommender systems~\citep{kroer2021computing,kroer2022market} or Internet advertising, we would not expect items to be drawn i.i.d.\ from a single distribution.
One alternative is to assume that data arrives adversarially. 
However, this leads to very pessimistic negative results and is not an accurate representation of the data one would expect to see in practice.
Instead, one would expect the data to have a strong stochastic component, but with changes over time, e.g., due to flow of traffic, breaking news events, or system updates~\citep{esfandiari2018allocation,balseiro2020best}. 

Motivated by the above considerations, we study online fair allocation when the data exhibits nonstationary behavior. In particular, we focus on the performance of the PACE algorithm of \citet{gao2021online}. We ask 
\vspace{-.2cm}
\begin{center}
\emph{
    How does PACE behave when nonstationarity is present in the stream of items?}
\end{center}
\vspace{-.2cm}

We show that, under several data-input models, the fairness and efficiency guarantees of the PACE algorithm are still preserved, up to errors due to the nonstationarity of the data input. In this sense, we significantly extend the main results in \citet{gao2021online}.
To show these results, we first consider the more general setting of nonstationary stochastic optimization and develop new performance guarantees for dual averaging in this setting.
Given the ubiquitous use of dual averaging in online and stochastic optimization, our results are of broader interest beyond (fair) resource allocation.

\subsection{Summary of Contributions}
\textbf{Novel convergence results for dual averaging under three nonstationary settings.}

We analyze the dual averaging (DA) algorithm for nonstationary stochastic optimization under different data input models, namely, (1) mildly corrupted, (2) ergodic and (3) periodic input data. Specifically, we consider the composite dual averaging algorithm, where the composite term is strongly convex.
We show that, in all cases, the iterates generated by dual averaging (DA) converge to the optimal solution in mean square, where the bound on the mean-square error decomposes into two terms: i) the typical $ O(\log t / t)$ guarantee known from the i.i.d.\ case, and ii)
a term that depends on the amount of nonstationarity in the data input model. Our results recover the classical bounds under i.i.d.\ data input as a special case. 

\textbf{Theoretical fairness and efficiency guarantees of PACE for nonstationary item arrivals.}

We consider the online fair allocation problem where item arrivals follow any of the three data input models that we consider for DA; these settings generalize the i.i.d.\ setting in \citet{gao2021online}.
Utilizing our convergence results for DA under nonstationary data input models, we show that, for item arrivals following these models, PACE ensures convergence of the pacing multipliers, again with a decomposition into a $ O(\log t / t)$ term as well as a term depending on the nonstationarity.
We then show that the agents' realized utilities, envy, regrets, and expenditures all obtain convergence bounds based on the convergence of pacing multipliers.
Our results show that PACE as an  online fair resource allocation algorithm is robust against distributional uncertainty of the input and automatically adapts to many different data input models without any parameter tuning.
In \cref{sec:experiments} we provide numerical experiments which corroborate the above theory and demonstrate the practical efficiency of PACE under different data input models.

An extensive review of related work is provided in \cref{app:relatedwork}.

\ifdefined \NIPS
\else
\subsection{Related Work}

Since our work studies competitive equilibrium computation, online fair resource allocation and stochastic optimization, while PACE employs the idea of pacing in auction mechanism design, we further discuss related work in these areas.

\paragraph{Convex optimization for computing competitive equilibria.}
Convex optimization algorithm (especially first-order methods) and their theory have been used to design and analyze algorithms for computing competitive equilibria, often through equilibrium-capturing convex programs \cite{birnbaum2011distributed,cole2017convex,cheung2020tatonnement,gao2020first,gao2021online}.
Applying a first-order method to such a convex program often leads to (recovers) interpretable market dynamics that emulate real-world economic behaviors, such as the proportional response dynamics \cite{birnbaum2011distributed,zhang2011proportional,cheung2018dynamics,gao2020first} and t{\^a}tonnement \citep{cheung2020tatonnement}.
The PACE algorithm of \citet{gao2021online} is no exception: it results from applying dual averaging to a specific convex program. 
Discrete variants of these convex programs have also been used for fair indivisible allocation~\citep{caragiannis2019unreasonable}, which yields some efficiency and fairness guarantees, though the discreteness breaks the connection to competitive equilibria.

\paragraph{(Online) fair resource allocation.} 
\citet{azar2016allocate,azar2010allocate} consider an online Fisher market with arbitrary item arrivals. 
They focus on a quality measure that is minimized at a competitive equilibrium and give an online algorithm that achieves a competitive ratio logarithmic in the size of the market and the ratio between the maximum and minimum (nonzero) buyer valuations over individual items. 
This algorithm requires solving a nontrivial linear program per iteration and is not known to improve with stochastic arrivals.
\citet{banerjee2022online} considers the problem of online allocation of divisible items to maximize Nash social welfare. 
They show that, under arbitrary item arrivals but with access to meaningful predictions of each buyer's total utility given all items, an online algorithm of the primal-dual type achieves a logarithmic competitive ratio. 
\citet{manshadi2021fair} studies the problem of rationing a social good and propose simple, implementable algorithms that promote fairness and efficiency. 
In their setting, it is the agents' demands rather than the supply that are sequentially realized and possibly correlated over time. 
\citet{bateni2021fair} uses Gaussian processes to model item arrivals and consider a budget-weighted proportional fairness metric. 
They propose a reoptimization policy that consumes buyers' budgets and clears the market gradually while ensuring a competitive ratio in hindsight w.r.t.\ this metric. 
This policy periodically resolves the Eisenberg-Gale (EG) convex program and does not require prior knowledge of future item arrivals. 
Our work differs from the above literature as follows. 
First, we consider practically-motivated nonstationary data input models for item arrivals that interpolate between fully adversarial and fully stochastic (i.i.d.).
Second, we show that the PACE algorithm, without any parameter tuning, adapts to different data input models and achieves strong performance guarantees that depend mildly on the ``nonstationarity'' of these models.
Given that PACE is scalable, interpretable and easy to implement this paper further ensures its effectiveness upon more realistic, non-i.i.d.\ item arrival processes.

\paragraph{(Nonstationary) stochastic optimization.}
Many stochastic optimization algorithms have been shown to attain nontrivial performance guarantees under under nonstationary data input \citep{duchi2012ergodic,balseiro2020best,besbes2015non}.
Motivated by high-dimensional and distributed optimization problems, \citet{duchi2012ergodic} analyzes stochastic mirror descent under ergodic data input.
\citet{balseiro2020best} analyzes a version of mirror descent for online resource allocation. 
They show that it achieve strong regret bounds under different data input models without knowing the model in advance.
The ergodic and periodic data input models in this paper are motivated by those considered in \citet{duchi2012ergodic} and \citet{balseiro2020best}. Different from these papers which focus on mirror descent, this paper focuses on the dual averaging algorithm, a different stochastic optimization algorithm particularly suitable for the equilibrium-capturing convex program we study.
Furthermore, we achieve stronger results than those past papers, by focusing on a setting where a composite term has strong convexity.

\paragraph{Pacing in auction mechanism design.}
The PACE algorithm uses first-price auctions with pacing.
As noted in \citet{gao2020first}, 
the idea of pacing has also been used widely in budget management strategies for Internet advertising auctions, with strong revenue and incentive guarantees (see, e.g., \citet{conitzer2019pacing,conitzer2021multiplicative,balseiro2019learning}). 
It is also used widely in practice, as reported in \citet{conitzer2021multiplicative}.
As shown in \citet{balseiro2020best}, pacing strategies ensure individual bidders' returns on their budgets and, if used by all buyers, lead to approximate Nash equilibria.
Similar to the analysis in \citet{gao2021online}, in this paper, we focus on competitive equilibrium and fairness properties of PACE, rather than game-theoretic (incentive) properties. 
\fi

\noindent\textbf{Notation.}
We use $1_t$ to denote the vector of ones of length $t$ 
and $e_j$ to denote the vector with one in the $j$-th entry and zeros in the others.
We use $\Delta(\Theta)$ to denote the space of probability measures on a measurable space $\Theta$,
and $\Delta_n$ to denote the simplex in $\R^n$. 
To measure the nonstationarity in the input data, we will use the total variation distance. Given two probability measures $P$ and $Q$, it is defined as
 $
 \| P-Q\|_\TV \defeq (1/2) \int | \frac{\diff P}{\diff \mu} - \frac{\diff Q}{\diff \mu} |\diff \mu \;,
 $ 
 where $\mu$ is a supporting measure. 
\section{Preliminaries on Online Fair Allocation}
An online fair allocation instance with infinitely divisible items with $n$ agents and a finite horizon $t$ consists of a tuple $\mathsf{A}=(n,t,\Theta, Q, v)$, where $\Theta$ is the (possibly uncountable) measurable space of all possible items, with an associated $\sigma$-algebra $\mathcal{M}$ and a probability measure $\mu$, the distribution $Q \in\Delta(\Theta^t)$ is the distribution over possible sequences of items  $\gamma = (\theta_1, \dots, \theta_t ) \in \Theta^t$,
each of unit supply, and the set $v = (v_1,\dots, v_n )\in L^1_+(\Theta)^n$ is the set of 
valuation functions of the $n$ agents. Here $ L^1_+(\Theta)$ is the space of positive integrable functions on $\Theta$. Agent 
$i$ sees a utility of $v_i(\theta)$ in item $\theta\in\Theta$. Abusing notation we let $v_i(\gamma) = \big(v_i (\theta^1), \dots, v_i(\theta^t)\big)$ denote the valuation for agent $i$ of items in the sequence $\gamma$. 
Let $Q^\tau$ be the marginal distribution of the item $\theta^\tau$ at time $\tau$ and $\bar Q = (1/t)\sumtau Q^\tau$.
We assume $\int _\Theta v_i \d \bar Q = 1$ for all $i\in[n]$. We further assume $\vinfty \defeq \max_i \| v_i \|_\infty< \infty$.
We stress that the PACE algorithm that we study is not going to require access to either the valuation functions $v$ or the set of possible items $\Theta$; these are only required in order to discuss the resulting bounds.

Given an instance $\mathsf{A}$, the decision maker allocates the stream of items $\gamma$ one at a time, in an irrevocable manner. 
At time $\tau$ when item $\theta_\tau$ is revealed, the decision maker must choose an allocation $x^\tau =(x^\tau_1, \dots, x^\tau_n) \in \Delta_n$ based on information available at that time, and allocate accordingly. Here the $i$-th entry of $x^\tau$ is the fraction of item 
$\theta^\tau$ allocated to agent $i$. On receiving her fraction, agent $i$ realizes a utility of 
$u^\tau_i \defeq v_i(\theta^\tau) x^\tau_i$.
We let $x = 
(x^1, \dots, x^t )$ denote the sequence of allocations made over time. For agent $i$, let $x_i = (x^1_i,\dots, x^t_i) \in \R^t$ denote the fraction of items given to agent $i$ across time. With this notation, the total utility of agent $i$ is 
$\langle x_i , v_i(\gamma) \rangle$.
The goal of the decision maker is to choose, in an online fashion, an allocation $x$ such that it achieves some form of both efficiency and fairness guarantees.

\subsection{Benchmark: The Hindsight Allocation}

As a benchmark, we will consider the hindsight-optimal allocation.
Suppose all items are presented to the decision maker as opposed to arriving one by one. 
In that case, a fair and efficient allocation can be found by allocating using the \emph{Eisenberg-Gale} (EG) convex program~\cite{eisenberg1959consensus}.
EG picks the allocation that maximizes the sum of weighted logarithmic utilities (which is equivalent to maximizing the weighted geometric mean of utilities):
\# \label{eq:EGprogram}
\max_{x\geq0, u\geq 0}
\Bigg\{ 
\sumiton    B_i  \log  (U_i)
\;\bigg|\;  
U_i   \leq  \big\langle v_i(\gamma), {x}_{i} \big\rangle 
  \;\; \forall i \in [n] 
\;,\;\;
    \sumiton {x}_{i}^\tau \leq 1 \;\; \forall \tau \in [t]
\Bigg\} \;.
\#
The weights $B_i$ represent the priority given to each agent, and they they can be interpreted as budgets in a market-based interpretation of the EG allocation.\footnote{The hindsight allocation \cref{eq:EGprogram} can be interpreted as a competitive equilibrium from equal incomes (CEEI) in the corresponding Fisher market;
see \cref{app:fishermarket} for more details on this interpretation.
}.
We will focus on the case where $B_i = t/n$ for all $i$, but all our results extend directly to the case of unequal weights, which can be useful in settings such as when buyers have quasilinear utilities~\cite{gao2021online,conitzer2019pacing} or when it is desirable to give a larger allocation to certain agents. 

The PACE algorithm asymptotically converges to the optimal dual solution, which is
\# \label{eq:dual-beta-finite}
\beta^\gamma\defeq  \argmin_{\beta \geq 0}
\Bigg\{  \frac1t \sumtau \max_{i\in[n]} \beta_i v_i(\theta^\tau)  - \frac1n \sumiton \log \beta_i 
\Bigg\} \;.
\#
We will also be interested in the \emph{underlying problem} implied by the average item supplies $s = \diff \bar Q / \diff \mu$.
Letting $\big\langle v_i, {x}_{i} \big\rangle \defeq \int_\Theta v_i x_i \diff \mu $,
this leads to the infinite-dimensional analogue of \eqref{eq:EGprogram}:
\# \label{eq:primal-infinite}
\max_{x \in L^\infty_+(\Theta), u\geq 0}
\Bigg\{ 
    \frac{1}{n} \sumiton \log  (u_i)
\;\bigg|\;  
u_i   \leq  \big\langle v_i, {x}_{i} \big\rangle 
  \;\; \forall i \in [n] 
,\;\;
    \sumiton {x}_{i} \leq s
\Bigg\} \;,
\#
We let $u^*$ denote the optimal utilities in \cref{eq:primal-infinite}.
 The infinite-dimensional analogue of the dual \eqref{eq:dual-beta-finite} is the following. For any $\delta_0 > 0$, 
The infinite-dimensional analogue of \eqref{eq:dual-beta-finite} is the following. For any $\delta_0 > 0$, 
\# \label{eq:dual-beta-infinite}
\beta^* \defeq \argmin_{\frac{1}{n(1+\delta_{0})} \leq \beta \leq 1+\delta_{0}}
\Bigg\{  \int_\Theta  \Big(\max_{i\in[n]} \beta_i v_i(\theta) \Big) \bar Q(\diff \theta)  
- \frac1n \sumiton \log \beta_i 
\Bigg\}
\;.
\#
A rigorous mathematical treatment of the infinite-dimensional program can be found in \cite{gao2020infinite} and \citep[Section 2]{gao2021online}. 
Note the additional constraint in \cref{eq:dual-beta-infinite} on $\beta$ does not affect the optimal solution since $1/n \leq \beta^*_i \leq 1$; see Lemma~1 in~\citet{gao2020infinite}. 

It is well-known that the hindsight allocation generated by the EG program enjoys the following efficiency and fairness properties:
\vspace{-.2cm}
\begin{enumerate}
    \item
    Pareto optimality: we cannot strictly increase any agent's utility without decreasing some other agents' utility.
    \item Envy-freeness: each agent prefers their own allocation to that of any other agent: $\langle v_i(\gam), x^*_i \rangle \geq \langle v_i(\gam), x^*_k \rangle$ for all $k\neq i$.
    \item Proportionality: every agent achieves at least as much utility as under the uniform allocation, i.e. $\langle v_i(\gam), x^*_i \rangle \geq \langle v_i(\gam),(1/n) 1_t \rangle$.
\end{enumerate}
\vspace{-.2cm}
Therefore the hindsight EG allocation is the gold standard that we assume the decision maker would use if she had known the sequence of items $\gamma$ in advance.
However, in the online setting the decision maker does not know this sequence, and must therefore instead attempt to approximate an equally good allocation in online fashion.

For an item sequence $\gam$, we let  $x^\gamma$ denote the optimal hindsight allocation, which is an optimal solution to \cref{eq:EGprogram},
and we denote the resulting total and average utility as
\vspace{-.1cm}\# \label{def:equtil}
U_{i}^\gamma \defeq
\big\langle v_i({\gamma}), x_i^\gamma \big\rangle \ =  \sum_{\tau=1}^t 
x_i^{\gamma, \tau} v_i(\theta^\tau) , \quad u_{i}^\gamma \defeq (1/t) \cdot U^\gamma_i\;.
\#
\subsection{Performance Metrics}

We measure the performance of an online allocation rule $x$ on the instance $\gamma$ via the following two quantities. The \emph{regret} of agent $i$ is the difference between the total hindsight equilibrium utility $ U_i^\gamma$ and their realized utilities $\utaui$ under $x$ 
\# \label{eq:defregret}
\Reg_{i,t} (\gamma) \defeq  U_i^\gamma-  \sumtau  u^\tau_i \;.
\#
The \emph{envy} of agent $i$ is the maximal extent to which they prefer the allocation of any other agent:
\# \label{eq:defenvy}
\Envy_{i,t}(\gamma) \defeq \max_{k\in [n]} \big\{ \langle v_i(\gamma), x_k\rangle - \langle v_i(\gamma), x_i\rangle  \big\}\;.
\#

We seek to understand the worst-case behavior of an algorithm when facing a certain class of input distributions. For a given input distribution class $\sfC \subset \Delta(\Theta^t)$, we will develop bounds on the worst-case regret and envy under any distribution in $\sfC$:
\$
\sup_{Q \in \sfC} \E_{\gamma \sim Q}\big[\Reg_{i,t}(\gamma)\big] \;,
\quad 
\sup_{Q \in \sfC} \E_{\gamma \sim Q}\big[\Envy_{i,t}(\gamma)\big] \;.
\$

\subsection{The PACE Algorithm} \label{sec:pace-description}
In this section, we review the PACE (\underline{P}ace \underline{A}ccording to \underline{C}urrent \underline{E}stimated Utility) dynamics~\citep{gao2021online}. 
In PACE, each item is allocated via first-price auction, and each agent constructs bids by scaling their value by a \emph{pacing multiplier}.
The pacing multipliers are maintained using simple, distributed updates that can be handled either by the agents or by the platform.

Algorithmic details are displayed in \cref{alg:pace}. Here $\Pi_{[\ell,h]}[x] = \max\{\ell, \min\{x, h \}\}$.
At every time step $\tau$ an item $\thetau$ is revealed. At that point every agent comes up with a \emph{bid} for that item, which is equal to their value for the item multiplied by their current pacing multiplier $\betataui$.
Then, the agents submit these bids to a first-price auction, and the item is allocated to the highest bidder. For concreteness, we choose the bidder with the smallest index if a tie occurs, but any rule works.
Each agent then observes their realized utility, updates their average utility received so far, and updates their pacing multiplier accordingly.
As pointed out in \citet{gao2021online}, PACE is an instantiation of dual averaging \cite{xiao2010dual} applied to the dual of the hindsight allocation program in \eqref{eq:dual-beta-infinite}.

PACE has many attributes desirable in real-world applications.

\textbf{Highlight 1. Decentralization. }
The PACE dynamics can be run in either centralized (by having the mechanism designer emulate the pacing process for each agent) or decentralized fashion (since the auction-based allocation is the only centralized step at each iteration), and are therefore suitable for Internet-scale online fair division and online Fisher market applications. 

\textbf{Highlight 2. Pure Allocation. }
PACE allows each item to be fully allocated to a single agent, even though the hindsight performance metric is allowed to utilize fractional allocations. While fractional allocations can be interpreted as randomized allocations in many large-scale settings, this may not always be desirable, for example when allocating food to food banks.

\textbf{Highlight 3. Tuning-free. } An important fact about the PACE dynamics is that each agent has no stepsize parameter whatsoever, which means that no stepsize tuning is required. Moreover, PACE is robust against \emph{the types} of item arrivals since the algorithm needs neither knowledge of the item distribution~$P$ nor the input type~$\sfC$.

\begin{algorithm}[t]
	\SetAlgoNoLine
	\KwIn{number of agents $n$, horizon $t$, algorithm parameter $\delta_0 > 0$.}
    \textbf{Initialize:} {Set~$\beta^1 = (1+\delta_0) \cdot 1_n$.}

    Environment draws the item sequence $\gamma=\{ \theta^1,\dots, \theta^t\}$ from the distribution $Q$. 

    \For{$\tau = 1,\dots, t$ when item $\thetau$ is revealed
    }{
            Agent $i$ bids $\betataui \vithetau$, the whole item $\thetau$ is allocated to the highest bidder $\itau$  (with arbitrary tie breaking)
            $
            \itau\defeq \min \{\argmax_{i\in[n]} \betataui \vithetau \} \;.
            $ 
            \label{line:sugbradient}

            Agent $i$ updates current average utility  
            $
                \utaui = \vithetau \indi \{ i=\itau \} \;,
                \quad
                \ubartaui = \frac{1}{\tau} \sum_{s=1}^\tau u^s_i \;.
            $ 
            \label{line:averageutil}
     
            Agent $i$ updates the pacing multiplier
            $
            \beta_{i}^{\tau+1}=\Pi_{[\ell, h]}\big[ 1 / (n\bar{u}_{i}^{\tau})\big]
            ,
            $ 
            \label{line:projection}
            where the interval $[\ell,h]=\Big[\tfrac{1}{(1+\delta_0)n}, 1+\delta_0\Big]$. 
            
        }
	\caption{PACE$(n,t,\delta_0)$}
    \label{alg:pace}
\end{algorithm}


In addition to the regret and the envy performance metrics, we will also derive results for the following two quantities that characterize the long-run behavior of PACE. 
Let $\bar u ^ t = (1/t)\cdot \sumtau u^\tau$ be the vector of average realized utilities for all agents. We will show that the agents' utilities converge to those associated to the \emph{underlying offline fair allocation problem}, $u^*$, defined in \cref{eq:primal-infinite}, in a mean-square sense, i.e.,
$
\E \big[\|\bar u^t - u^*  \|\sq \big] \to 0,
$
as long as the error due to nonstationarity grows sublinearly in the number of time periods.
Secondly, define the \emph{expenditure} of agent $i$ at time $\tau$ by 
$
b_{i}^{\tau} \defeq \beta_{i}^{\tau} v_{i}\left(\thetau\right) \indi\left\{i=i^{\tau}\right\}.
$
We will show $(1/t)\cdot \sumtau b^\tau_i \to 1/n$ in mean square as well, 
as long as the error due to nonstationarity grows sublinearly in the number of time periods.

\section{Main Results} \label{sec:input-model}
This section introduces the main results of this paper: the behavior of PACE under three different types of nonstationary input models. 
All prior results on fair online allocation have been either for worst-case inputs (with much more conservative guarantees and not for the PACE algorithm)~\citep{azar2016allocate,banerjee2022online}
or for i.i.d. input data~\cite{gao2021online}.

We first introduce some notation that will be useful for describing these input models.
For $s > \tau \geq 1$ let $Q^s(\theta^{1:\tau})$ denote the conditional distribution of $\theta^s$ given $\{\theta^{1},\dots, \theta^{\tau}\}$. 
For a subset $I$ of $[t]$ let $Q^I$ denote the joint distribution of the variables $\{ \theta^\tau\}_{\tau \in I}$.
Let $\bar Q = (1/t)\cdot \sumtau Q^\tau$ be the uniform mixture of $\{ Q^\tau\}_\tau$. 
We study three types of input: independent input with adversarial corruption, ergodic and Markov input, and periodic input.
For each input setting, we describe our main theorem for the performance guarantees of PACE here. The proofs are given in \cref{sec:mechanism-pace}, because these results rely on developing a theory of nonstationary performance of DA, which is done in \cref{sec:main-nonstat-da}.

\subsection{Independent Input with Adversarial Corruption}

Adversarial perturbation of a fixed item distribution models scenarios where the items generally behave in a predictable manner, but for  some time steps the input behaves erratically. Typically this is assumed to happen only for a small number of time steps. Such perturbation could be malicious, 
for example when item arrivals are manipulated in favor of certain agents; 
or non-malicious, such as unpredictable surges of certain keywords on search engines~\citep{esfandiari2018allocation}.

We study a type of adversarial perturbation where the item distribution at each time step might be corrupted by an arbitrary amount, but distributions at different time steps are independent of each other. We assume the average corruption is bounded by $\delta$, as measured in TV distance. The set of distributions over sequences that we consider is then:
\#
\sfC^{\mathrm{ID}}(\delta) \defeq  \bigg\{ {{Q}} \in \Delta({\Theta})^{t}: \frac1t \sumtau \| Q^\tau -  \bar Q \|_\TV \leq \delta \bigg\}
\;.
\#
We use $\tilde O$ to hide numeric constants and polynomials of $n$, $\vinfty$, and $\log t$.
Our main fair online allocation result for the adversarial corruption case is:
\begin{theorem} [Independent Case]
    \label{thm:main-envy-regret-indep}
    For the adversarially corrupted and independent case, \cref{alg:pace} guarantees 
    that 
    for an instance $\sfA = (n,t,\Theta,Q,v)$, 
    we have
    \#
     &\sup_{Q \in \sfC^\mathrm{ID}(\delta)} \E_{\gamma \sim Q}\big[\Reg_{i,t}(\gamma)\big] \;,
     \sup_{Q \in \sfC^\mathrm{ID}(\delta)} \E_{\gamma \sim Q}\big[\Envy_{i,t}(\gamma)\big] 
      = \tilde O \big( \sqrt{t}  + \sqrt{\delta} \cdot t\big)
    \#
    and 
    \$
    \scriptstyle{\sup_{Q \in \sfC^\mathrm{ID}(\delta)} \E_{\gamma \sim Q} \big [\| \bar b^t - (1/n)1_n \|\sq \big]
    \;,
    \sup_{Q \in \sfC^\mathrm{ID}(\delta)} \E_{\gamma \sim Q} \big [\| \bar u ^t - u^* \|\sq \big] 
    \;, \sup_{Q \in \sfC^\mathrm{ID}(\delta)} \E_{\gamma \sim Q} \big [\| \bar u ^t - u^\gamma \|\sq \big] }
    = \tilde O(\delta + 1/t) \;.
    \$
\end{theorem}
The result shows that the regret and 
envy performance metrics degrade linearly in the average 
corruption $\delta$. In the i.i.d.\ case where $\delta = 0$, we recover the
$\sqrt{t}$ regret rate and the $1/t$ rate of convergence for utilities and expenditures in terms of the mean-square error from \cite{gao2021online}.
If out of the $t$ distributions of items in each time step only $O(\sqrt{t})$ are corrupted, each by a constant amount, then the $\sqrt{t}$ regret and envy bounds, as well as $1/t$ convergence rates, are also preserved.

\subsection{Ergodic Input and Markov Processes}
To handle correlation across time, we next study ergodic inputs. For these inputs, strong correlation might be present for items sampled at nearby time steps, but the correlation between items decays as they are separated in time. For any integer $\iota$ such that $1 \leq \iota \leq t-1$, we measure the $\iota$-step deviation from some distribution $\Pi\in\Delta(\Theta)$ by the quantity 
\$
\delta(\iota) \defeq \sup_{\gamma} \sup_{\tau = 1,\dots, t-\iota} \| Q^{\tau + \iota} (\theta^{1:\tau})- \Pi \|_\TV
\;.
\$
Intuitively, this definition tells us that, no matter where and when we start the item arrival process, it takes only $\iota$ steps to get $\delta(\iota)$-close to the distribution $\Pi$.
We will consider the set of ergodic input distributions whose $\iota$-step deviation is bounded by $\delta$:
\begin{align}
\hspace{-.2cm}
    \sfC^{\mathrm{E}}(\delta, \iota) \defeq \bigg\{ Q \in \Delta(\Theta^t):\sup_{\gamma} \sup_{\tau = 1,\dots, t-\iota} \| Q^{\tau + \iota} (\theta^{1:\tau})- \Pi \|_\TV \leq \delta, \text{ for some } \Pi\in\Delta(\Theta) \bigg\}
\;.
\end{align}
\begin{theorem} [Ergodic Case]
    \label{thm:main-envy-regret-ergodic}
    For the ergodic case, \cref{alg:pace} guarantees that for an instance $\sfA = (n,t,\Theta,Q,v)$, we have
\#
    \sup_{Q \in \sfC^\mathrm{E}(\delta, \iota)} \E_{\gamma \sim Q}\big[\Reg_{i,t}(\gamma)\big]
    \;,
    \sup_{Q \in \sfC^\mathrm{E}(\delta, \iota)} \E_{\gamma \sim Q}\big[\Envy_{i,t}(\gamma)\big]
    & =
    \tilde O ( \sqrt{ \iota t} + \sqrt{\delta} \cdot t )
\#
and
\$ 
  \scriptstyle{  \sup_{Q \in \sfC^\mathrm{E}(\delta, \iota)} \E_{\gamma \sim Q} \big [\| \bar b^t - (1/n)1_n \|\sq \big]
    \;,
    \sup_{Q \in \sfC^\mathrm{E}(\delta, \iota)} \E_{\gamma \sim Q} \big [\| \bar u ^t - u^* \|\sq \big] 
    \;,
    \sup_{Q \in \sfC^\mathrm{E}(\delta, \iota)} \E_{\gamma \sim Q} \big [\| \bar u ^t - u^\gam \|\sq \big] 
    = 
    \tilde O( \delta +  \iota /t ) 
    \;.}
\$
\end{theorem}
\begin{remark}[Markov Input]\label{rm:markov-input}
    
    We can specialize the result in \cref{thm:main-envy-regret-ergodic} to fast mixing or Markov item sequences.  Fast mixing means the deviation $\delta$ decreases exponentially, i.e., for all $1 \leq \iota \leq t-1$, 
    \#
    \sup_{\gamma} \sup_{\tau = 1,\dots, t-\iota} \| Q^{\tau + \iota} (\theta^{1:\tau})- \Pi \|_\TV \leq M \rho^{\iota}
    \;,
    \#
    for some $M>0$, $\rho \in [0,1)$, and $\Pi$ is the stationary distribution.
    Examples include finite state-space time-homogeneous Markov chains and uniformly ergodic Markov chains on general state spaces~\citep[Chapter 16]{meyn2012markov}.
    In these cases, setting
        $
            \iota = \frac{\log t  + \log M }{\log (\rho\inv)} = O(\frac{\log t}{\log(\rho\inv )} )
        $ 
        implies $\delta \leq 1/t$.
        This means the Markov chain from which $\gamma$ is generated takes $O(\log t)$ steps to get $(1/t)$-close to stationarity.
        The dominant term for the regret in \cref{thm:main-envy-regret-ergodic}  (further ignoring $M$) is then
        $  ( 1+\frac{1}{\log(\rho \inv)} ) ^{1/2}  \sqrt{ t} .$ 
        The term in the parenthesis reflects the inflation caused by input dependency. To recover the case of i.i.d.\ input, we simply send $\rho \to 0$ and the usual $\sqrt{t}$ regret and envy rates and $1/t$ utility and expenditure convergence rates are again recovered.
    \end{remark}

\subsection{Periodic Input}

Item sequences often exhibit statistical periodic structure. For example, when allocating compute time to requestors, there will be more requests during weekdays and less on weekends. The compute request patterns vary throughout the week, and yet the weekly pattern would repeat over time. 
Similarly, Internet traffic typically exhibits periodic structure.

Formally, assume that the horizon $t$ divides into $K$ blocks of time, each of size $q$.
This divides the item sequence $\gam$ into consecutive blocks of length $q$. 
Within each block, we allow a distribution with arbitrary dependence between time steps, but we assume that the blocks, as a whole, are identically and independently distributed. We define the set of periodic input distributions as follows:
\#
\mathsf{C}^{\mathrm{P}}(q) \defeq \big\{{Q} \in \Delta(\Theta^{q})^K : Q^{1: q}=Q^{q+1: 2 q}=\ldots=Q^{t-q+1: t} \big\} 
\;.
\#
\begin{theorem} [Periodic Case]
\label{thm:main-envy-regret-periodic}
 For the periodic case, \cref{alg:pace} guarantees that for an instance $\sfA = (n,t,\Theta,Q,v)$, we have
\#
\sup_{Q \in \sfC^\mathrm{P}(q)} \E_{\gamma \sim Q}\big[\Reg_{i,t}(\gamma)\big]
\;,
\sup_{Q \in \sfC^\mathrm{P}(q)} \E_{\gamma \sim Q}\big[\Envy_{i,t}(\gamma)\big]
& = 
\tilde O \big( \sqrt{qt} \big)
\#
and 
\$
\scriptstyle{ \sup_{Q \in \sfC^\mathrm{P}(q)} \E_{\gamma \sim Q} \big [\| \bar b^t - (1/n)1_n \|\sq \big]
 \;,
 \sup_{Q \in \sfC^\mathrm{P}(q)} \E_{\gamma \sim Q} \big [\| \bar u ^t - u^* \|\sq \big] 
 \;,
 \sup_{Q \in \sfC^\mathrm{P}(q)} \E_{\gamma \sim Q} \big [\| \bar u ^t - u^\gam \|\sq \big]}
= \tilde O( q\sq / t ) \;.
\$
\end{theorem}

If the length $q$ of the blocks is of order $o(t)$ then the time-averaged regret and envy are both vanishing. For the i.i.d.\ case, we can set $q=1$ to recover the previous results.

Now we comment on dependence on the period length $q$. Suppose the item sequence consists of $K$ blocks, and blocks are i.i.d.\ We still allow arbitrary dependence within a block. The proof of \cref{thm:main-envy-regret-periodic} essentially relies on the result (\cref{thm:maintext-da-periodic}) that DA produces iterates whose squared error is of order $O(q^2/t)$. 
Consider dual averaging with the knowledge of the block structure $q$. Then the rate $1/K = q/t$ can be achieved by executing DA using one randomly chosen data point within a block, throwing away the rest in that same block. Such selection produces $K$ i.i.d. samples from the distribution. In comparison, the rate in $O(q^2/t)$ is worse off by a factor of $q$ due to not knowing the block-structure information.

\section{Proof Technique: Nonstationary Dual Averaging}
\label{sec:main-nonstat-da}

The PACE dynamics can be cast as online dual averaging \citep{xiao2010dual} applied to the dual of the hindsight allocation program in \eqref{eq:EGprogram}. 
This will be discussed in more detail in \cref{sec:pace-as-da} and \cref{app:pace-as-da}. 
However, in order to characterize its performance under various types of nonstationary input, we first extend the general convergence results for dual averaging to incorporate nonstationary input. The convergence results that we developed for dual averaging under three different types of input models are novel and are of independent interest. 

We remark that the results for the stochastic setting given by \cite{xiao2010dual} cannot be used directly, since they rely on the stringent i.i.d.\ assumption. 
\cite{duchi2012ergodic} consider ergodic mirror descent (MD) for convex problems under some (but not all) of our nonstationary input models.
Their results and analysis cannot be used in our case either, since their results do not allow using the composite structure, whose strong convexity we leverage.
Moreover, unlike DA, an MD-based approach would require tuning parameters such as stepsizes.

In this section, after introducing the setup of DA in \cref{sec:setup-of-da}, we present a DA convergence result for independent but not identical input in \cref{sec:da-indept-with-proof}, for which we outline the proof idea and technical challenges in \cref{sec:da-indept-with-proof}. Due to space limitations, we present DA convergence results for ergodic and periodic inputs in \cref{sec:da-ergodic-and-periodic}.

In the nonstationary setup of DA, we emphasize that whenever we mention convergence, we mean convergence of DA iterates to the population-level optimum $w^*_\Pi$ (or sometimes the hindsight optimum), up to some error caused by nonstationarity.

\subsection{Review: The DA Algorithm}
\label{sec:setup-of-da}

We review the dual averaging setup in the strongly convex case \cite[\S 1.1]{xiao2010dual}. Consider a stochastic optimization problem of the form 
\# \label{eq:linproblem}
    \min_w \bigg\{\phi(w)  \defeq  \E_{z\sim\Pi} \big[F(w,z)\big] = \E_{z\sim \Pi}\big[f(w,z)\big] + \Psi(w) \bigg\} \;,
\#
where $w\in (\R^d, \|\cdot\|)$ is the variable, $\Psi$ is a closed convex function with closed domain $\dom \Psi := \{ w\in \R^n: \Psi(w)<\infty \}$.  The expectation is taken over a probability distribution $\Pi$ on a measurable space $Z$. 
For each $z\in Z$, the function $ f(\cdot, z)$ is convex and subdifferentiable (a subgradient always exists) on $\dom \Psi$. Let $F(w,z) = f(z,w) + \Psi(w)$.

\textbf{Assumptions.} Let $\sfG(w, z) $ be a fixed element in the set of subgradients $\partial_w f(w, z)$. 
\vspace{-.2cm}
\begin{enumerate}
  \item 
  For almost every $z$, it holds $\|\sfG(w,z)\| _* \leq G $, where $\| \cdot\|_* =\max _{\|w\| \leq 1}\langle s, w\rangle$ is the dual norm.
  \item 
  There exists an $\barF\in \mathbb R$ such that $ F(w,z) \leq \barF$ for all $w\in \dom \Psi$ and (almost every) $z$. 
  \item 
  $\Psi$ is $\sigma$-strongly convex, i.e., $\Psi(\alpha w+(1-\alpha) u) \leq \alpha \Psi(w)+(1-\alpha) \Psi(u)-\frac{\sigma}{2} \alpha(1-\alpha)\|w-u\|^{2}$ for $w,u\in \dom \Psi$. 
\end{enumerate}
\vspace{-.2cm}
  Because of our strong convexity assumption, the solution to \eqref{eq:linproblem} is unique. Associated with $\Pi$ we define 
  $
  w^*_\Pi \defeq \argmin \E_{z\sim\Pi} \big[F(w,z)\big]
  \;.
  $

The dual averaging algorithm (DA) \cite[Algorithm 1]{xiao2010dual} aims to produce a sequence converging to the optimal point $w^*_\Pi$ or minimize the associated regret \citep[\S1.2]{xiao2010dual}. The algorithmic details for DA are presented in \cref{alg:da}. Note that although \cite{xiao2010dual} only considers the case of i.i.d.\ data, DA iterates are defined for every input sequence $\zts$, regardless of any distributional properties of the sequence.

The first step in our analysis is a relationship between regret and the suboptimality $\|w_t - w^*_\Pi\|$ derived by \citet{xiao2010dual}. 
Consider the dual averaging algorithm with data $\zts$. 
We denote the one-step subgradient by $g_\tau
\defeq  \sfG(w_\tau, z_\tau)$
and the average subgradient by $\gbartau = (\sum_{s = 1}^\tau g_{s}) /\tau$.
Given data $\zts$, we define the regret and the sum of squared subgradient norms~\footnote{
  See the first equation on page 2584 in \cite{xiao2010dual}. In \cite{xiao2010dual}'s notation, set $\beta_\tau = 0$ all $\tau \geq 1$ and $\beta_0 = \sigma$, plug in the bound $h(w_2) \leq 2\|g_1\|_*\sq/\sigma$ and we have the expression of $\Delta_t$ in our paper.
}
\$
R_{t}(w) 
\defeq \sum_{\tau=1}^{t} \big(F(w_\tau, z_\tau) - F(w, z_\tau)\big) 
\;,\quad
\Delta_{t}  
\defeq
\frac{1}{2\sigma} \big( 5 \| g_1\|_*^2 + 
\sum_{\tau=1}^{t-1}{\| g _{\tau+1}\|_*^2}/{\tau}\big) 
.
\$
The bound $\Delta_t \leq {(6+\log t )G\sq}/({2\sigma})$ holds 
in a deterministic manner due to the bounded subgradient assumption. 

\citet{xiao2010dual} shows the following bound on suboptimality of $w_t$
\begin{lemma}[Regret Bound, Section B.2 in \cite{xiao2010dual}] \label{fact:linxiao-intext}
  For any sequence $\zts$, any $w \in \dom \Psi$, any $t=1,2,\dots$, it holds $ \left\|w_{t+1}-w\right\|^{2} \leq \frac{2}{\sigma t} \big(\Delta_{t}-R_{t}(w)\big) \; . $
\end{lemma}

\subsection{Review: PACE as Dual Averaging} \label{sec:pace-as-da}

In this section we review how to cast PACE as dual averaging applied to the problem \eqref{eq:dual-beta-infinite}. This derivation was originally given in \citet{gao2021online}. Let $f(\beta, \theta)= \max_i \beta_i v_i(\theta)$, $\Psi(\beta) = - \frac1n \sumiton \log\beta_i$, in which case we get 
$
F(\beta, \theta) = f(\beta,\theta) + \Psi(\beta) = \max_{i\in[n]} \big\{\beta_i v_i(\theta) \big\} - \frac1n \sumiton \log \beta_i\;.
$
Following \cite[\S 5]{gao2020infinite}, since $f(\,\cdot\,,\theta)$ is a piecewise linear function, a subgradient is 
$
\sfG ( \beta, \theta) := v_{i^\tau}(\theta) e_{i^\tau}  \in \partial_\beta f(\beta,\theta)
\;,
$ 
where $i^\tau = \min\{ \argmax_i \beta_i v_i(\theta)\}$ is the index of the winning agent (see, e.g., \citet[Theorem 3.50]{beck2017first}).
Based on this instantiation of DA, we get that the iterates $\{ \betatau\}_{\tau = 1}^{t+1}$ generated by the PACE dynamics (\cref{alg:pace}) are exactly the iterates $w_\tau$ generated by $\mathrm{DA}(\sfG,\Psi,\gamma)$ (\cref{alg:da}). A proof is given in \cref{app:pace-as-da}.

\begin{algorithm}[t]
	\SetAlgoNoLine
	\KwIn{
    subgradient $\sfG$, 
    regularizer $\Psi$ and
    data $\zts$
    .}
    \textbf{Initialize:} set $\bar g_0 = 0$ and $w_1 = \argmin \Psi$.

    \For{$\tau = 1,\dots, t$
    }{
      Observe $z_\tau$ and compute 
      $g_{\tau} = \sfG(w_{\tau},z_{\tau})$.
      \label{item:da-compute-g(t)}
      
      Average subgradients (the \emph{dual average}) via 
      $\bar{g}_{\tau} = \frac{\tau-1}{\tau}\bar{g}_{\tau-1} + \frac{1}{\tau}{g}_{\tau}$.
      \label{item:da-update-g-bar} 
      
      Compute the next iterate $w_{\tau+1} = \argmin_w \{ \langle \bar{g}_{\tau}, w\rangle + \Psi(w) \}$. 
      \label{item:da-update-w}
    }
	\caption{DA$(\sfG,\Psi, \zts)$}
  \label{alg:da}
\end{algorithm}

Before moving on to showing our results, let us first touch on the fact that dual averaging provides worst-case regret guarantees. Naively, one may expect that these regret guarantees would directly translate into regret guarantees on the primal performance, meaning a bound on $\sup_{\gamma} \Reg_{i,t}(\gamma)$.
However, such dual regret bounds do not imply a worst-case primal regret bound.
From a technical perspective, it is unclear how a regret bound on the dual objective would translate to a regret bound on envy or utilities.
Secondly, based on results in the online fair allocation literature~\citep{azar2016allocate, banerjee2022online}, it is known that not only can we not get a no-regret guarantee on the primal performance, it is not even possible to achieve a constant competitive ratio in the worst-case setting.


\subsection{Independent Data with Adversarial Corruption} \label{sec:da-indept-with-proof}
 
We present a DA convergence result with independent data in this section. Theorem statements for the ergodic case and periodic case are presented in \cref{sec:da-ergodic-and-periodic}.
Discarding the i.i.d.\ assumption on the data $\zts$, we let $P$ be the joint distribution of $\{z_\tau \}_\tau$ and let $P^\tau$ be the marginal distribution of $z_\tau$. 
We study the relationship between the DA iterates $w_{t+1}$ and $w^*_\Pi$ by bounding the mean-square difference 
$\E_{\zts \sim P }\big[\|w_{t+1} - w^*_\Pi \|\sq\big]$, 
and thus demonstrate in what sense the data distribution $P$ should stay close to the i.i.d.\ distribution $\Pi$ in order to preserve DA convergence. 

We first introduce a variant of $\sfC^{\mathrm{ID}}(\delta)$ with a target distribution $\Pi$: 
\# \label{eq:da-indep-dist-class}
\sfC^{\mathrm{ID}}(\delta; \Pi ) \defeq  \bigg\{ {{P}} \in \Delta({\Theta})^{t}: \frac1t \sumtau \| P^\tau -  \Pi \|_\TV \leq \delta \bigg\} \;.
\#
\begin{theorem}[DA Convergence, Independent Data with Adversarial Corruption] \label{thm:maintext-da-indep}
If $\zts \sim P$ and $P\in\sfC^{\mathrm{ID}}(\delta,\Pi)$. Then for $t\geq1$,
  \$
  \E_{\zts \sim P} \big[\| w_{t+1}-w^*_\Pi \|\sq\big]
  \leq 
  \frac{(6 + \log t) G^2 }{\sigma^2 t} + \frac{8\bar F }{\sigma }\delta = \tilde O(\delta + 1/t ) \;.
  \$
Moreover, the rate $ \tilde O(\delta + 1/t )$ applies to $\E\big[ \| w_{t+1} - w^*_\gamma\|_2^2 \big]$ and $\E\big[ \| w^*_\gam - w^*_\Pi \|_2^2 \big]$ (See \cref{sec:deduction-to-hindsight-optima}).
\end{theorem}

\begin{remark}
  From this result we can tell when DA retains last-iterate convergence. Suppose the number of corrupted data points is of order $o(t)$ (assuming corruption on each data is of the same order), then the corruption per item is $\delta = o(1)$ and DA converges to the optimal solution as $t\to\infty$. Furthermore, if the corruption per data point is of order $\delta = O(1/t)$, then the fast rate $1/t$ is retained.
\end{remark}

By \cref{sec:pace-as-da}, we get as a corollary that the PACE iterates $\{\beta_t \}$ will converge to $\beta^*$, the solution to the infinite-dimensional dual program \eqref{eq:dual-beta-infinite}. This is the building block for our results in \cref{sec:input-model}.  

While the full proof is too long to fit in the paper, we now give a proof sketch to show the main ideas behind how we derive \cref{thm:maintext-da-indep}.
We begin the proof by noticing \cref{fact:linxiao-intext} is deterministic and valid for any $\zts$. 
Now set $w =  w^*_\Pi $ in \cref{fact:linxiao-intext}.
If the input data $\zts$ were i.i.d.\ from $\Pi$, i.e., $P=\Pi^{\otimes t}$, then $\E[R_t(w^*_\Pi)]$ would be greater than zero, and we would obtain 
\$
  \E \big[\|w_{t+1}-w^*_\Pi \|^{2} \big] 
  \leq 
  \frac{2}{\sigma t} \Big(\E[\Delta_t]  - \E\big[R_t(\wst)\big]\Big) 
  \leq 
  \frac{ (6+\log t)G\sq}{\sigma\sq t} \;.
\$

However, in the nonstationary case, the regret $\E[R_t(w^*_\Pi)]$ might be negative. At a high level, our results are achieved by  introducing appropriate measures of the nonstationarity and then lower bounding $\E[R_t(w^*_\Pi)]$ based on those measures.

To this end, we decompose the regret as follows.
    Write
    \$
      R_t({w^*_\Pi}) = &
      \underbrace{\sum_{\tau=1}^{t} \big(F(w_\tau, z_\tau) - {\phi_\Pi}(w_\tau) \big)
     + \sum_{\tau=1}^{t} \big({\phi_\Pi}({w^*_\Pi}) - F({w^*_\Pi}, z_\tau) \big)}_{\text{I}}
      + 
      \underbrace{\sum_{\tau=1}^{t} \big({\phi_\Pi}(w_\tau) - {\phi_\Pi}({w^*_\Pi})\big)}_{\text{II}} \;.
    \$
    By optimality of ${w^*_\Pi}$ we have $\II \geq 0$. Using the bound on the TV distance between $\{ P^\tau \}_\tau$ and $\Pi$, and boundedness of $F$ we can control the other two terms. The key is, conditional on $\cF_{\tau-1}$, the iterate $\wtau$ is deterministic and the distribution of $\ztau$ is $P^\tau$ due to independence assumption. For each term in the first summation, we condition on $\cF_{\tau-1}$ and obtain
    \$
| \E[ F(w_\tau, z_\tau) - {\phi_\Pi}(w_\tau) |\cFtaum]   | 
& = \bigg | \E\bigg[ \int_\cZ F(w_\tau, z) \Ptaugiventaumdz - \int_\cZ F(w_\tau, z)\diff \Pi (z)|\cFtaum\bigg]\bigg |
\\
& = \bigg | \E\bigg[ \int_\cZ F(w_\tau, z) P^\tau(\diff z) - \int_\cZ F(w_\tau, z)\diff \Pi (z)|\cFtaum\bigg]\bigg |
\\
& \leq  \E\Bigg[ \bigg | \int_\cZ  F(w_\tau, z)P^\tau(\diff z) - \int_\cZ F(w_\tau, z)\diff \Pi (z)\bigg | |\cFtaum\Bigg]
\\ 
& \leq 2 \bar F   \|P^\tau - \Pi \|_\TV \;.
\$
The second sum can be handled similarly. For the detailed proof and a generalization see \cref{sec:arb} in \cref{app:non-stat-da}.

  We now discuss how this paper handles nonstationarity differently from the existing literature.
  Let us use \citet{balseiro2020best} as a reference point, since they consider very similar categories of nonstationary inputs. From a online constrained optimization perspective, \citet{balseiro2020best} relies on the fact that their objective is separable across timesteps in order to handle nonstationarity. In contrast, this structure does not exist for the online fair allocation problem, because the objective takes the logarithm of the utility over time.
Concretely, in \citet{balseiro2020best}, the objective is of the form $\sum_{\tau=1}^t f_\tau(x)$. This type of time-separability occurs in all works that we are aware of for non-stationary inputs, also e.g. the ergodic mirror descent paper \citet{duchi2012ergodic} (\citet{duchi2012ergodic} is not an online setting, but the expectation in their objective is analogous to time separability). Time separability is used as a key property for deriving  regret bounds in those papers. The separability enables translating dual regret to primal regret by a weak duality argument (see e.g. Prop.~1 in \citet{balseiro2020best} where a time-separable dual allows weak duality).
  
  In contrast, our problem has time-separability only in the dual formulation, but not in the primal one, which is where we ultimately want guarantees (since we are interested in utilities converging). Our contribution to handling nonstationarity in online fair allocation is showing that a dual approach works even without the separability structure.  We begin by deriving convergence guarantees on the last dual iterate, which we achieve by modifying the dual averaging proof to take into account nonstationarity and analyzing the stability of the dual variables in dual averaging. We note that this technique depends on strong convexity, unlike for the time-separable case.
  Next, to go from dual last-iterate convergence to primal regret bounds we use the first-order optimality condition for EG program, which are specific techniques for our problem (such techniques were also used in the previous PACE paper \citet{gao2021online}).
    
\section{Conclusion}
We established new convergence results for dual averaging under nonstationary data input models, namely, adversarial corruption, ergodic, and block-independent input models. 
Leveraging these results, we showed that, for online fair allocation problems with item arrivals generated from the above nonstationary data input models, the PACE algorithm automatically adapts to them and achieves asymptotic fairness and efficiency without any parameter tuning. 
Numerical experiments demonstrated the effectiveness of PACE under these data input models.

\ack{
This research was supported by the Office of Naval Research Young Investigator
Program under grant N00014-22-1-2530, and the National Science Foundation award
IIS-2147361.
}
\newpage
\appendix
\bibliographystyle{apalike}
\bibliography{refs.bib}
\newpage
\section*{Checklist}

The checklist follows the references.  Please
read the checklist guidelines carefully for information on how to answer these
questions.  For each question, change the default \answerTODO{} to \answerYes{},
\answerNo{}, or \answerNA{}.  You are strongly encouraged to include a {\bf
justification to your answer}, either by referencing the appropriate section of
your paper or providing a brief inline description.  For example:
\begin{itemize}
  \item Did you include the license to the code and datasets? \answerYes{See Section~\ref{gen_inst}.}
  \item Did you include the license to the code and datasets? \answerNo{The code and the data are proprietary.}
  \item Did you include the license to the code and datasets? \answerNA{}
\end{itemize}
Please do not modify the questions and only use the provided macros for your
answers.  Note that the Checklist section does not count towards the page
limit.  In your paper, please delete this instructions block and only keep the
Checklist section heading above along with the questions/answers below.

\begin{enumerate}

\item For all authors...
\begin{enumerate}
  \item Do the main claims made in the abstract and introduction accurately reflect the paper's contributions and scope?
    \answerYes{See \cref{sec:input-model}}
  \item Did you describe the limitations of your work?
    \answerNo{Our work is mostly theoretical and studies the fair-allocation algorithm PACE under various nonstationary settings.}
  \item Did you discuss any potential negative societal impacts of your work?
    \answerNo{Not applicable.}
  \item Have you read the ethics review guidelines and ensured that your paper conforms to them?
    \answerYes{We read them.}
\end{enumerate}

\item If you are including theoretical results...
\begin{enumerate}
  \item Did you state the full set of assumptions of all theoretical results?
    \answerYes{See \cref{sec:input-model} and \cref{sec:main-nonstat-da}.}
        \item Did you include complete proofs of all theoretical results?
    \answerYes{See \cref{sec:pace_as_da}, \cref{app:non-stat-da} and \cref{sec:proof_for_pace}.}
\end{enumerate}

\item If you ran experiments...
\begin{enumerate}
  \item Did you include the code, data, and instructions needed to reproduce the main experimental results (either in the supplemental material or as a URL)?
    \answerYes{See \cref{sec:experiments}.}
  \item Did you specify all the training details (e.g., data splits, hyperparameters, how they were chosen)?
    \answerYes{See \cref{sec:experiments}.}
        \item Did you report error bars (e.g., with respect to the random seed after running experiments multiple times)?
    \answerNo{No applicable.}
        \item Did you include the total amount of compute and the type of resources used (e.g., type of GPUs, internal cluster, or cloud provider)?
    \answerYes{See \cref{sec:experiments}.}
\end{enumerate}

\item If you are using existing assets (e.g., code, data, models) or curating/releasing new assets...
\begin{enumerate}
  \item If your work uses existing assets, did you cite the creators?
    \answerNo{We use synthetic data.}
  \item Did you mention the license of the assets?
  \answerNo{N/A.}
  \item Did you include any new assets either in the supplemental material or as a URL?
  \answerNo{N/A.}
  \item Did you discuss whether and how consent was obtained from people whose data you're using/curating?
  \answerNo{N/A.}
  \item Did you discuss whether the data you are using/curating contains personally identifiable information or offensive content?
  \answerNo{N/A.}
\end{enumerate}

\item If you used crowdsourcing or conducted research with human subjects...
\begin{enumerate}
  \item Did you include the full text of instructions given to participants and screenshots, if applicable?
  \answerNo{N/A.}
  \item Did you describe any potential participant risks, with links to Institutional Review Board (IRB) approvals, if applicable?
  \answerNo{N/A.}
  \item Did you include the estimated hourly wage paid to participants and the total amount spent on participant compensation?
  \answerNo{N/A.}
\end{enumerate}

\end{enumerate}

\section{Related Work} \label{app:relatedwork}

Since our work studies competitive equilibrium computation, online fair resource allocation and stochastic optimization, while PACE employs the idea of pacing in auction mechanism design, we further discuss related work in these areas.

\paragraph{Convex optimization for computing competitive equilibria.}
Convex optimization algorithm (especially first-order methods) and their theory have been used to design and analyze algorithms for computing competitive equilibria, often through equilibrium-capturing convex programs \cite{birnbaum2011distributed,cole2017convex,cheung2020tatonnement,gao2020first,gao2021online}.
Applying a first-order method to such a convex program often leads to (recovers) interpretable market dynamics that emulate real-world economic behaviors, such as the proportional response dynamics \cite{birnbaum2011distributed,zhang2011proportional,cheung2018dynamics,gao2020first} and t{\^a}tonnement \citep{cheung2020tatonnement}.
The PACE algorithm of \citet{gao2021online} is no exception: it results from applying dual averaging to a specific convex program. 
Discrete variants of these convex programs have also been used for fair indivisible allocation~\citep{caragiannis2019unreasonable}, which yields some efficiency and fairness guarantees, though the discreteness breaks the connection to competitive equilibria.

\paragraph{(Online) fair resource allocation.} 
\citet{azar2016allocate} consider an online Fisher market with arbitrary item arrivals. 
They focus on a quality measure that is minimized at a competitive equilibrium and give an online algorithm that achieves a competitive ratio logarithmic in the size of the market and the ratio between the maximum and minimum (nonzero) buyer valuations over individual items. 
This algorithm requires solving a nontrivial linear program per iteration and is not known to improve with stochastic arrivals.
\citet{banerjee2022online} considers the problem of online allocation of divisible items to maximize Nash social welfare. 
They show that, under arbitrary item arrivals but with access to meaningful predictions of each buyer's total utility given all items, an online algorithm of the primal-dual type achieves a logarithmic competitive ratio. 
\citet{gkatzelis2021fair} study the setting where items arrive online and with two agents. They focus on satisfying the no-envy condition while maximizing social welfare, and show that one do this approximately by allocating items proportionally to valuations, assuming that valuations are normalized.
\citet{manshadi2021fair} studies the problem of rationing a social good and propose simple, implementable algorithms that promote fairness and efficiency. 
In their setting, it is the agents' demands rather than the supply that are sequentially realized and possibly correlated over time. 
\citet{bateni2021fair} uses Gaussian processes to model item arrivals and consider a budget-weighted proportional fairness metric. 
They propose a reoptimization policy that consumes buyers' budgets and clears the market gradually while ensuring a competitive ratio in hindsight w.r.t.\ this metric. 
This policy periodically resolves the Eisenberg-Gale (EG) convex program and does not require prior knowledge of future item arrivals. 
Our work differs from the above literature as follows. 
First, we consider practically-motivated nonstationary data input models for item arrivals that interpolate between fully adversarial and fully stochastic (i.i.d.).
Second, we show that the PACE algorithm, without any parameter tuning, adapts to different data input models and achieves strong performance guarantees that depend mildly on the ``nonstationarity'' of these models.
Given that PACE is scalable, interpretable and easy to implement this paper further ensures its effectiveness upon more realistic, non-i.i.d.\ item arrival processes.

\paragraph{(Nonstationary) stochastic optimization.}
Many stochastic optimization algorithms have been shown to attain nontrivial performance guarantees under under nonstationary data input \citep{duchi2012ergodic,balseiro2020best,besbes2015non}.
Motivated by high-dimensional and distributed optimization problems, \citet{duchi2012ergodic} analyzes stochastic mirror descent under ergodic data input.
\citet{balseiro2020best} analyzes a version of mirror descent for online resource allocation. 
They show that it achieve strong regret bounds under different data input models without knowing the model in advance.
The ergodic and periodic data input models in this paper are motivated by those considered in \citet{duchi2012ergodic} and \citet{balseiro2020best}. Different from these papers which focus on mirror descent, this paper focuses on the dual averaging algorithm, a different stochastic optimization algorithm particularly suitable for the equilibrium-capturing convex program we study.
Furthermore, we achieve stronger results than those past papers, by focusing on a setting where a composite term has strong convexity.

\paragraph{Pacing in auction mechanism design.}
The PACE algorithm uses first-price auctions with pacing.
As noted in \citet{gao2020first}, 
the idea of pacing has also been used widely in budget management strategies for Internet advertising auctions, with strong revenue and incentive guarantees (see, e.g., \citet{conitzer2019pacing,conitzer2021multiplicative,balseiro2019learning}). 
It is also used widely in practice, as reported in \citet{conitzer2021multiplicative}.
As shown in \citet{balseiro2020best}, pacing strategies ensure individual bidders' returns on their budgets and, if used by all buyers, lead to approximate Nash equilibria.
Similar to the analysis in \citet{gao2021online}, in this paper, we focus on competitive equilibrium and fairness properties of PACE, rather than game-theoretic (incentive) properties.

 \section{Review of Linear Fisher Market}  \label{app:fishermarket}

In fair division the goal is to perform this allocation in a \emph{fair} way, while simultaneously also guaranteeing some form of efficiency, typically Pareto efficiency.
In the case of allocating $m$ divisible goods to $n$ agents, the \emph{competitive equilibrium from equal incomes} (CEEI) allocation guarantees many fairness properties. In CEEI, every agent is endowed with a unit budget of faux currency, a competitive equilibrium is computed, i.e., a set of item prices along with an allocation that clears the market, and the resulting allocation is used as the fair allocation~\citep{varian1974equity}. 
This guarantees several fairness desiderata such as \emph{envy-freeness} (every person prefers their own bundle to that of any other person), \emph{proportionality} (every person prefers their own bundle over receiving their fair share $1/n$ of every item), and Pareto optimality (we cannot make any person better off without making at least one other person worse off).

A linear Fisher market refers to the tuple $\sfF=(n, m, B, {\mathsf{v}})$. The market consists of $n$ buyers and $m$ items. We assume each buyer has a budget of $B_i$. We use $\{1,\dots, m\}$ to represent the set of items, each of unit supply. The matrix ${\mathsf{v}} = ({\mathsf{v}}_1,\dots, {\mathsf{v}}_n )\in (\R^m_+)^n$ consists of valuations, with ${\mathsf{v}}_{i}^{j}$ being the valuation of item $j$ from buyer $i$. For buyer $i$, an {allocation} of items, $\mathsf{x}_i \in \R^m_+$, gives a utility of $ u_i(\mathsf{x}_i) := \langle {\mathsf{v}}_i, \mathsf{x}_i \rangle := \sum_{j=1}^m {\mathsf{v}}_{i}^{j} \mathsf{x}_{i}^{j}$. Note we use different fonts to distinguish notions that appear in both the online allocation problem and the Fisher market.
 
 \begin{defn}[Demand] \label{defn:demand-set}
     Given item prices $p \in \R^m_+$, the {demand} of buyer $i$ is its set of utility-maximizing allocations given the prices and budget:
     \# \label{def:demand}
     D_i (p) := \argmax \{  \langle {\mathsf{v}}_i, \mathsf{x}_i \rangle : \mathsf{x}_i \geq 0,\,  \langle p,  \mathsf{x}_i\rangle \leq B_i \}
     \;.
     \#
 \end{defn}
 
\begin{defn}[Market Equilibrium]
    The {market equilibrium} of $\mathsf{F}=(n,m,B,{\mathsf{v}})$ is an allocation-price pair $(\mathsf{x}^*, p^*) \in (\R^m_+)^n \times \R^m_+$ such that the following holds.
    \begin{enumerate}
        \item Supply feasibility: $\sumiton \mathsf{x}^*_i \leq {1}_m$. 
        \item Buyer optimality: $\mathsf{x}^*_i \in D_i (p^*)$ for all $i$.
        \item Market clearance: $\langle p^*, {1}_m - \sumiton \mathsf{x}^*_i \rangle = 0$.
    \end{enumerate}
    \label{defn:me-static}
\end{defn}

Market equilibrium and fair allocation are related as follows.
In CEEI, we construct a mechanism for fair division by giving each agent the same budget of fake currency, i.e., $B_i = B_j$ for all $i,j$, computing what is called a market equilibrium under this new market, and using the corresponding allocation as our fair allocation rule.

It is known that an allocation $\mathsf{x}^*$ from the set of CEEI has many desirable properties.
It is Pareto optimal (every market equilibrium is Pareto optimal by the first welfare theorem). 
It has no envy: since each agent has
the same budget in CEEI and every agent is buying something in their demand set, no envy must be satisfied, since they can afford the bundle of any other agent. Finally, proportionality is satisfied, since each agent can afford the bundle where they get $1/n$ of each good.

The ME is essentially a collection of optimization problems (\cref{def:demand}) coupled through the constraint~$\sumiton \mathsf{x}_i \leq \mathbf{1}_m$. A celebrated result is the Eisenberg-Gale convex program, which provides an equivalent characterization of ME.
 \#  \label{def:fisherEG}
     \max _{\mathsf{x}_1, \dots, \mathsf{x}_n } \;  \sumiton B_i \log \left\langle {\mathsf{v}}_{i}, \mathsf{x}_{i}\right\rangle
     \quad\mathrm{s.t.} \quad 
     \sumiton \mathsf{x}_{i}^{j} \leq 1 \;\; \forall j \in [m]\;,
      \quad\;\; 
      \mathsf{x}_i \in \R^m_+ \;\; \forall i \in [n]\;.
 \#
That is, we maximize the sum of logarithmic utilities under the supply constraint. The solution to the primal problem $\mathsf{x}^*=(\mathsf{x}_1^*,\dots, \mathsf{x}_m^*)$ along with the vector of dual variables $p^*$ yields a market equilibrium.

The hindsight allocation 
\cref{eq:EGprogram} is just the EG program of the linear Fisher market $\sfF_\sfA = (n,t,B, \sfv)$ where entries of the valuation matrix $\sfv$ are defined by $\sfv_{i}^{j} = v_i(\theta^j)$ for all $i\in[n],j\in[t]$, and $B = (t/n )1_n$.

\section{PACE as Dual Averaging} \label{sec:pace_as_da}

In this section, we show how to cast PACE as dual averaging. 
To this end, we will introduce infinite-dimensional Eisenberg-Gale-type convex programs for the allocation of a (possibly infinite/continuous) set of items. 
Here, the item supplies correspond to the probability density $\diff \bar{Q}/\diff\mu$ of the average item arrival distribution $\bar{Q}$. 
They serve as intermediate ``reference'' convex programs that facilitate the use of DA convergence results developed in the previous section to analyze PACE.
When the item space is continuous, the supply function, allocation rules, and the price function in these convex programs are (measurable) functions over such item spaces, which can be infinite-dimensional objects.
When item arrivals are drawn from a (fixed) distribution with density $s$, they correspond to the EG convex programs of the ``underlying market'' with item supplies $s$. 
Note that when the item space $\Theta$ is finite, the infinite-dimensional analogues reduce to the classical finite-dimensional EG convex programs.
After introducing these concepts we will show that our results on nonstationary DA allows us to derive comparable results on various PACE performance metrics.
The results of \citet{gao2021online} cast PACE as dual averaging for the EG convex program of the underlying market, and show guarantees with respect to that program. Here, we will show our results for that setting, as well as for the hindsight allocation problem.

Different from the presentation of DA in \citet{xiao2010dual}, we do not need an additional regularizer.
This paper focuses on DA for strongly convex problems with an existing regularizer in the loss (and hence no auxiliary regularizer is needed), since this is the setting used in the design and analysis of the PACE algorithm; similar convergence results under our new input models can be derived for the general form of DA given in \citet[Algorithm 1]{xiao2010dual} with an auxiliary regularizer for non-strongly convex problems. 

\label{sec:mechanism-pace}  
\subsection{The Dual of EG and the Infinite-Dimensional Analogue}
\label{sec:dual-and-inf-program}
We derive the dual program of \eqref{eq:EGprogram}. Introduce the dual variables $\beta_i\geq 0$ with $i\in[n]$ for each constraint of the first type and variables $p^\tau \geq 0$ with $\tau \in [t]$ for constraints of the second type. The Lagrangian $L: \R^{n\times t}_+ \times \R_+^n \times \R^n_+ \times \R^t_+ \to \R$ is given by
\$
L(x, U, \beta, p)
& = \sumiton B_i \log U_{i}+\sumiton \beta_{i}\big(\langle v_{i}(\gamma), x_{i}\rangle-u_{i}\big)+\sumtau p^{\tau} \bigg( 1 -\sumiton x_{i}^\tau\bigg)
\\
& = \sumtau p^\tau + \sumiton \bigg(B_i \log U_i - \beta_i u_i\bigg) + \sumiton \big\langle \beta_i v_i(\gam)-p,x_i \big\rangle
\;.
\$
Maximizing out the variables $(x,U)$ gives the dual program
\$
\min _{p \geq 0, \beta \geq 0}
\Bigg\{  \sumtau p^\tau - \sumiton  B_i \log  \beta_{i}  + \sumiton \big(B_i \log B_i - B_i\big)
\;\bigg|\;
p \geq \beta_i  v_i(\gamma) \;\; \forall i\in [n]
\Bigg\}\;.
\$
Dividing by $t$ (recalling $B_i = t/n$), ignoring constants and moving the constraint $p\geq \beta_i v_i(\gamma)$ to the loss, we obtain the following equivalent optimization problem:
\# 
\min_{\beta \geq 0}
\Bigg\{  \frac1t \sumtau \max_{i\in[n]} \beta_i v_i(\theta^\tau)  - \frac1n \sumiton \log \beta_i 
\Bigg\} \;.
\#
Let $\beta^\gam$ be the optimal solution. To recover the corresponding optimal $p^\gam$ we define
$p^{\gam,\tau}= \max_{i} \beta^\gam_i \vithetau$ for $\tau \in [t]$.

We recall the infinite-dimensional analogue of \eqref{eq:EGprogram}:
\# 
\max_{x \in L^\infty_+(\Theta), u\geq 0}
\Bigg\{ 
    \frac{1}{n} \sumiton \log  (u_i)
\;\bigg|\;  
u_i   \leq  \big\langle v_i, {x}_{i} \big\rangle 
  \;\; \forall i \in [n] 
,\;\;
    \sumiton {x}_{i} \leq s
\Bigg\} \;,
\#
where $s = \diff \bar Q / \diff \mu$ is the average item supply function, and $\big\langle v_i, {x}_{i} \big\rangle \defeq \int_\Theta v_i x_i \diff \mu $. The infinite-dimensional analogue of \eqref{eq:dual-beta-finite} is the following. For any $\delta_0 > 0$, 
\# 
\min_{\beta \geq 0}
\Bigg\{  \int_\Theta  \Big(\max_{i\in[n]} \beta_i v_i(\theta) \Big) \bar Q(\diff \theta)  
- \frac1n \sumiton \log \beta_i 
\;\bigg|\;  
\frac{1}{n(1+\delta_{0})} \leq \beta_i \leq 1+\delta_{0}  \;\; \forall i\in[n] 
\Bigg\}
\;,
\#
A rigorous mathematical treatment of the two infinite-dimensional programs can be found in \cite{gao2020infinite} and \citep[Section 2]{gao2021online}. 
Note the additional constraint in \eqref{eq:dual-beta-infinite} on $\beta$ does not affect the optimal solution since $1/n \leq \beta^*_i \leq 1$; see Lemma~1 in~\citet{gao2020infinite}.

The relationship between the finite and infinite versions of  \eqref{eq:dual-beta-finite} is that we have replaced the uniform averaging in \eqref{eq:dual-beta-finite} with an integral w.r.t.\ the item average distribution $\bar Q$ in \eqref{eq:dual-beta-infinite}. 
For notational simplicity, we suppress dependence on $\bar Q$ and let $(x^*, u^*, \beta^*)$ denote the optimal solutions to the infinite-dimensional programs \eqref{eq:primal-infinite} and \eqref{eq:dual-beta-infinite}. Define the corresponding optimal $p^*$ in \eqref{eq:dual-beta-infinite}  by $p^* \defeq   \max_{i\in[n]} \beta^*_i v_i$. 

\subsection{PACE as Dual Averaging}\label{app:pace-as-da}

In this section we review how to cast PACE as dual averaging applied to the problem \eqref{eq:dual-beta-infinite}. This derivation was originally given in \citet{gao2021online}. Recall $f(\beta, \theta)= \max_i \beta_i v_i(\theta)$, $\Psi(\beta) = - \frac1n \sumiton \log\beta_i$ and then 
\$
F(\beta, \theta) = f(\beta,\theta) + \Psi(\beta) = \max_{i\in[n]} \big\{\beta_i v_i(\theta) \big\} - \frac1n \sumiton \log \beta_i\;.
\$
Following \cite[\S 5]{gao2020infinite}, since $f(\,\cdot\,,\theta)$ is a piecewise linear function, a subgradient is 
\$
\sfG ( \beta, \theta) := v_{i^\tau}(\theta) e_{i^\tau}  \in \partial_\beta f(\beta,\theta)
\;,
\$ 
where $i^\tau = \min\{ \argmax_i \beta_i v_i(\theta)\}$ is the index of the winning agent (see, e.g., \cite[Theorem 3.50]{beck2017first}).

\begin{lemma}[PACE as Dual Averaging]\label{lm:paceasda}
  The iterates $\{ \betatau\}_{\tau = 1}^{t+1}$ generated by the PACE dynamics (\cref{alg:pace}) are exactly $\mathrm{DA}(\sfG,\Psi,\gamma)$ (\cref{alg:da}).
\end{lemma}

\begin{proof}[Proof of \cref{lm:paceasda}]
Interpret the DA updates using the following substitution: $\Theta \leftrightarrow Z$, $\thetau \leftrightarrow z_{\tau}$, $ \beta^{\tau+1} \leftrightarrow w_{\tau+1}$ and $\bar{g}_{\tau,i}  \leftrightarrow \bar{u}^\tau_i$. For initialization in DA, choose  $\beta^1$ to be the minimizer of $\Psi$ over the cube $[(1+\delta_n)\inv n\inv 1_n, (1+\delta_0) 1_n]$ and set $\bar{g}_0 = \bar u ^0 = 0$.

(1) Subgradient computation $\Leftrightarrow$ choose the winning bidder (\cref{line:sugbradient} of PACE). 

(2) Average subgradient $\Leftrightarrow$ update current averaged utilities (\cref{line:averageutil} of PACE). The $i$-th entry of $\sfG ( \beta,\theta)$ is exactly the time-$\tau$ realized utility of agent $i$ in PACE, that is, $g_{\tau,i} = v_i(\theta^\tau) \indi\{ i=i^\tau \} = u^\tau_i$. Then the average gradient, $\bar{g}_{\tau} = \frac{\tau-1}{\tau}\bar{g}_{\tau-1} + \frac{1}{\tau}g_{\tau}$, is the same as the time-averaged utilities:
\$ \bar{g}_{\tau,i} = \frac{\tau-1}{\tau} \bar{g}_{\tau-1,i} + \frac{1}{\tau} v_i(\thetau) \indi\{ i=i^\tau \} \;.\$

(3) Solve regularized problem $\Leftrightarrow$ update pacing multiplier (\cref{line:projection} of PACE). The minimization problem is separable in agent index $i$ and exhibits a simple and explicit solution. Recall $\bar{g}_{\tau,i} = \bar{u}^\tau_i$: 
    \$ \beta^{\tau+1}_i = 
    \argmin
    \left\{\bar{g}_{\tau,i} \beta_i - \frac1n \log \beta_i 
    \Big|\,
    \frac{1}{n(1+\delta_0)}\leq \beta_i  \leq 1+\delta_0 \right\}  \ \Rightarrow\ \beta^{\tau+1}_i = \Pi_{ [\ell, h] } \left(\frac{1}{n\bar{u}^\tau_{i}}\right)\;.
    \$

\end{proof}

\subsection{Performance Guarantees via Dual Averaging} \label{sec:beta-convergence-in-text}

Now that we have cast PACE as an instantiation of dual averaging and developed results for convergence in nonstationary settings, the following theorems follow easily from the general convergence results for DA. 
Recall the hindsight optimum $\beta^\gam$ is defined in \eqref{eq:dual-beta-finite}, and its infinite-dimensional counterpart $\beta\st$ is defined in \eqref{eq:dual-beta-infinite}.

\begin{theorem}[Convergence of PACE, the Independent Case] 
    \label{thm:maintext-beta-conv-indep}
    Assume the item sequence $\gamma \sim Q$ and $Q\in\sfC^{\mathrm{ID}}(\delta)$. Choose $\delta_0 = 1$ in PACE. 
    It holds for $t\geq 1$, 
    \# \label{eq:beta-conv-indep-1}
 \E \big[\| \beta^t - \betast\|\sq\big] \leq \frac{(6+\log t) n\sq \vinftysq }{t} + 8n\vinfty \cdot \delta = \tilde O (\delta + 1/t)
    \;.
    \#
    Moreover, the rate $ \tilde O(\delta + 1/t )$ applies to $ \E[\| \beta^\gam - \betast\|\sq]$ and  $\E[\| \beta^{t+1} - \beta^\gam  \|\sq]$.
\end{theorem}
\begin{proof}
Set $P = Q$,  $\Pi = \bar Q$, $\sigma=1/n$, and $\bar F = \vinfty$ in \cref{thm:maintext-da-indep}.
\end{proof}

Here the convergence of $\beta^t$ to the hindsight counterpart $\beta^\gam$ is of practical importance. This is because $\beta^\gam$ can always be computed after the fact, while its infinite-dimensional counterpart $\beta^*$ is not necessarily obtainable.

\begin{theorem}[Convergence of PACE, the Ergodic and the Periodic Cases] 
    \label{thm:maintext-beta-conv-ergodic-periodc}
    For the erogdic case, i.e., $\gamma \sim Q$ and $Q\in \sfC^{\mathrm{E}}(\delta, \iota)$, it holds for $t \geq 1$, 
    \$ 
    \E[\|\beta^{t+1} - \betast \|\sq ] \leq \frac{C_{E,1} + C_{E,2} \cdot \iota} {t} + C_{E,3} \cdot \delta = \tilde O\bigg( \delta + \frac{\iota }{t}\bigg)
    \;.
    \$
    where
    $C_{E,1} = n\sq\vinftysq \big(6+\log t\big) $,
    $C_{E,2} = 4n \big(\vinftysq(1+\log t) + \vinfty \big) $ and
    $C_{E,3} = 8n\vinfty $.

    For the periodic case, i.e., $\gamma \sim Q$ and $Q\in \mathsf{C}^{\mathrm{P}}(q) $, it holds for $t \geq 1$
    \$ 
    \E[\|\beta^{t+1} - \betast \|\sq ] \leq   \frac{ C_{P,1} + C_{P,2} \cdot q\sq}{t}
    = \tilde O ( q\sq / t  )\;.
    \$
    where $C_{P,1} = C_{E,1}$ and $C_{P,2}  = 2n \vinftysq(1+\log t) $.

    For both cases, similar convergence results can be stated for $\E[\| \beta^\gam - \betast\|\sq]$ and $\E[\| \beta^{t+1} - \beta^\gam\|\sq]$ and are omitted here.
    \end{theorem}
    \begin{proof}
        For the first inequality, set $P = Q$,  $\Pi = \bar Q$, $\sigma=1/n$, and $\bar F = \vinfty$ in \cref{thm:maintext-da-ergodic}.  
        For the second inequality, set additionally $\delta^\mathrm{b} = 0$ and $|\cP|_\infty = q$ in
        \cref{thm:maintext-da-periodic}.
    \end{proof}

\subsection{From Dual EG Performance Bounds to Primal Performance Bounds}

Convergence of $\beta^\tau$ to $\beta\st$ implies the convergence of the average utilities and expenditure to their infinite-dimensional counterparts.
This follows almost directly from results developed by \citet{gao2021online}. In particular, they show:
\begin{lemma}[PACE Long-Run Behavior, \cite{gao2021online}]
    \label{lm:dual-to-expendicture-util}
    For any distribution $Q \in \Delta(\Theta^t)$, let $\gam \sim Q$.
    It holds for $t\geq 1$,
    \$
     \E \big[\| \bar b^t - (1/n)1_n \|\sq \big] \leq 2  \E[\|  \betatp - \betast\|\sq ] + 4\vinftysq 
    \bigg(\frac1t \sumtau \E[\| \betatau - \betast \|\sq]\bigg)
    \;,
    \$
    and
    \$
     \E \big[\| \bar u ^t - u^* \|\sq \big] \leq C_u \cdot \E[\|  \betatp - \betast\|\sq ]
    \;,
    \$
    where $ C_u = n\sq\left(\vinftysq/ \delta_{0}\sq +\left(1+\delta_{0}\right)^{2}\right)$.
\end{lemma}

Finally, we relate results in \cref{sec:beta-convergence-in-text} to the main quantities of interest: regret and envy, defined in \eqref{eq:defregret} and \eqref{eq:defenvy}, as well as convergence to the hindsight utilities. Similar results are given by \citet{gao2021online}, and our proof is almost identical to theirs, simply extended to the nonstationary case as well as to the hindsight allocation problem.

\begin{lemma}[Regret and Envy] \label{lm:dual-to-regret}
For any distribution $Q \in \Delta(\Theta^t)$, let $\gam \sim Q$.
It holds for $t\geq 1$,
\$
\E [ \| \bar u ^t - u^\gam \|\sq ] \leq   C_{r,1} \cdot \E[ \| \betatp - \betast \|\sq ] + n \cdot C_{r,2} \cdot \bigg(\frac1t \sumtau \E[\| \betatau - \betast \|\sq]\bigg)\;,
\$
\$
\E \big[\Reg_{i,t}(\gam)\big] \leq t \cdot \sqrt{ 
    C_{r,1} \cdot \E[ \| \betatp - \betast \|\sq ] + C_{r,2} \cdot 
    \bigg(\frac1t \sumtau \E[\| \betatau - \betast \|\sq]\bigg)
}\;,
\$
and 
\$
\E \big[\Envy_{i,t}(\gam)\big] \leq   t \cdot \sqrt{C_{e,1}\cdot \E[ \| \betatp - \betast\|\sq ] + C_{e,2} \cdot \bigg(\frac1t \sumtau \E[\| \betatau - \betast \|\sq]\bigg)} \;,
\$
where
$C_{r,1} = 2 C_u$, 
$C_{r,2} = 2 n\sq \vinftysq$, 
$C_{e,1}= 2 (1+ n \sq )C_u$ and 
$C_{e,2} = 4 \vinftysq n\sq + 2n^3$.
\end{lemma}

Now we have all the ingredients to prove the convergence of PACE.

\begin{proof}[Proof of \cref{thm:main-envy-regret-indep}, \cref{thm:main-envy-regret-ergodic} and \cref{thm:main-envy-regret-periodic}]
    Combine \cref{lm:dual-to-expendicture-util} with \cref{thm:maintext-beta-conv-indep} and \cref{thm:maintext-beta-conv-ergodic-periodc} and we obtain the first set of inequalities in \cref{thm:main-envy-regret-indep,thm:main-envy-regret-ergodic,thm:main-envy-regret-periodic}. 
    Then combine \cref{lm:dual-to-regret} with \cref{thm:maintext-beta-conv-indep} and \cref{thm:maintext-beta-conv-ergodic-periodc}, and we obtain the second set of inequalities.
     
\end{proof}

\section{Proofs for Nonstationary Dual Averaging}
\label{app:non-stat-da}

The DA algorithm \cref{alg:da} is obtained in \citet[\S 3.2]{xiao2010dual} for the strongly convex case. We simply set $h=(1/\sigma) \Psi$, $\beta_0 = \sigma$ and $\beta_t= 0$ for $t\geq 1$ in \cite[Algorithm 1]{xiao2010dual}. 

Recall $P^\tau$ is the distribution of $z_\tau$.
For integers $\tau$ and $\tau'$ ($\tau' \geq \tau$), let $P^{\tau'}(\cdot \cond z_{1:\tau})$ denote the distribution of $z_{\tau'}$ if the process starts at $z_{1:\tau} = \{z_1, \dots, z_\tau\}$.

\subsection{Remark: Convergence to the Hindsight Optimum} \label{sec:deduction-to-hindsight-optima}
Before developing the nonstationary convergence theory, we digress a bit and introduce a simple deduction through which we can easily show the convergence of DA iterates to the optimum of the hindsight problem based on convergence to $w^*_\Pi$.  
Given data $\zts$, define the sum $\phi_\gamma (w) = (1/t) \cdot \sumtau F(w,\ztau) $ and its 
unique minimizer $w^*_\gamma =\argmin \phi_\gamma(w)$.
We claim all results developed for $\|w_{t+1}-w^*_\Pi \|^{2}$ will also hold for the hindsight suboptimality $\|w_{t+1}-w^*_\gam \|^{2}$.

Note the following inequality:
  \$
    R_{t}(w^*_\gamma) 
    = \sum_{\tau=1}^{t} \big(F(w_\tau, z_\tau) - F(w^*_\gamma, z_\tau)\big) 
    = \sum_{\tau=1}^{t} \big(F(w_\tau, z_\tau) - \phi_\gamma(w^*_\gamma)\big) 
    \geq \sum_{\tau=1}^{t} \big(F(w_\tau, z_\tau) - \phi_\gamma(w^*_\Pi) \big) \,,
  \$
the last term being exactly $R_t(w^*_\Pi)$. Choose $w = w^*_\gamma$ in \cref{fact:linxiao-intext} and we obtain
  \$
    \E\big[ \| w_{t+1} - w^*_\gamma\|_2^2 \big] 
    \leq \frac{1}{\sigma t} \Big(\E[\Delta_t]  - \E\big[R_t(w^*_\gamma)\big]\Big)
    \leq \frac{1}{\sigma t} \Big(\E[\Delta_t]  - \E\big[R_t(\wst)\big]\Big)
    \;.
  \$
It follows that all lower bounds for the regret $\E\big[R_t(\wst)\big]$ can be turned into an upper bound for the hindsight suboptimality measure $\|w_{t+1}-w^*_\Pi \|^{2}$. Convergence to the hindsight optimum is of practical importance since the hindsight optimum $w^*_\gam$ can typically be computed, where this is not always the case for the population optimum~$w^*_\Pi$.

A simple consequence of the deduction above is the following. 
We note the inequality 
\$
\|w^*_\gam -w^*_\Pi \|^{2} \leq 2 \|w_{t+1}-w^*_\Pi \|^{2} + 2 \|w_{t+1}- w^*_\gam  \|^{2}
\;.
\$
Therefore convergence results for $ \|w_{t+1}-w^*_\Pi \|$ and $\|w_{t+1}- w^*_\gam  \|^{2}$ similarly hold for $\|w^*_\gam -w^*_\Pi \|^{2}$.


\subsection{Proof for Adversarial Corruption and Independent Data} \label{sec:arb}

Assume the data $\gamma = \zts$ follows the distribution $P$ with no further assumptions. We let $z_{1:0} = \emptyset$ and further let $P^\tau(\cdot\given z_{1:0}) = P^\tau$, the marginal distribution of $z_\tau$.

Define the progressive deviation from $\Pi$
    \# \label{def:deltat}
      \delta_t \defeq \sup_{z_1,\dots,\zt} \sumtau \| \Ptaugiventaum - \Pi \|_\TV  \; . 
    \# 
    If the data are independent, then it holds
    \$ 
      \delta_t = \sumtau \| P^\tau - \Pi \| _ \TV \;.
    \$
Note for the independent case, $\delta_t = t\cdot \delta$ where $\delta$ is in \cref{eq:da-indep-dist-class}.
If the data further has identical distribution $\Pi$ then $\delta_t=0$.
If $\delta_t = O(\log t)$ we call data has mild corruption.

\begin{theorem}[Corrupted Data, Generalization of \cref{thm:maintext-da-indep}]  \label{thm:corrupted}
It holds
    \#
    \E\big[ \| w_{t+1} - {w^*_\Pi}\|_2^2 \big] 
    = O\Big(\frac{\log t}{\sigma^2 t} + \frac{\delta_t}{\sigma  t}\Big) 
    \;, 
    \#
where $O$ hides dependence on constants, $G$ and $\barF$. Recall $\sigma$ is the strong convexity parameter of $\Psi$.
\end{theorem}

\begin{remark} 
  In either the i.i.d.\ case or the mild corruption case ($\delta_t = O(\log t)$), we recover the usual $O(\log t/t)$ rate.
\end{remark}
 
\emph{Proof of \cref{thm:corrupted}}
In \cref{fact:linxiao-intext}, the term $\Delta_t$ is upper bounded in a deterministic manner. So it remains to handle $R_t$. In the i.i.d.\ case, $\E R_t ({w^*_\Pi})$ is positive and thus can be dropped:
\$ \E[R_{t}({w^*_\Pi}) ]      
=  \E \Big[\sum_{\tau=1}^{t} \big(F(w_\tau, z_\tau) - F({w^*_\Pi}, z_\tau) \big)\Big] 
= \sum_{\tau=1}^{t} \big({\phi_\Pi}(w_\tau) - {\phi_\Pi}({w^*_\Pi}) \big) \geq 0 \;. \$
However, to handle corrupted data, we need to use

\begin{lemma} \label{lm:mildcorruption}
  The regret can be lower bounded by the corruption parameter $\delta$:
\$\E[R_t({w^*_\Pi})] \geq  - 4 \cdot \bar F \delta_t \;.\$
\end{lemma}

Plugging in the above lemma, we get
\# \label{eq:concrete-indep-bound}
\E[\left\|w_{t+1}-w\right\|^{2} ] \leq \frac{1}{\sigma t} \big(\E[\Delta_t]  - \E[R_t(\wst)]\big) \leq 
\frac{\big(6+\log t\big)G\sq}{\sigma} + \frac{8 \bar F}{\sigma } \delta = 
O\Big(\frac{G^2 \log t}{\sigma^2 t} + \frac{\bar F\delta_t}{\sigma t}\Big) \;.
\# 
Thus, to complete the proof of \cref{thm:corrupted} we only need to prove \cref{lm:mildcorruption}.

\begin{proof}[Proof of \cref{lm:mildcorruption}]
Write
\$
  R_t({w^*_\Pi}) = &
  \sum_{\tau=1}^{t} \big(F(w_\tau, z_\tau) - {\phi_\Pi}(w_\tau) \big)\tag{I}
  \\
 &+  \sum_{\tau=1}^{t} \big({\phi_\Pi}({w^*_\Pi}) - F({w^*_\Pi}, z_\tau) \big)\tag{II}
  \\
  &+ 
  \sum_{\tau=1}^{t} \big({\phi_\Pi}(w_\tau) - {\phi_\Pi}({w^*_\Pi})\big) \;. \tag{III}
\$
By optimality of ${w^*_\Pi}$ we have $\III \geq 0$. 

\underline{Bounding $\I$ and $\II$.} Conditional on $\cF_\tau$, the iterate $\wtau$ is deterministic and the distribution of $\ztau$ is $\Ptaugiventaum$.

Note
\$ 
\E[ F(w_\tau, z_\tau) - {\phi_\Pi}(w_\tau) ] 
= \E[ \E[ F(w_\tau, z_\tau) - {\phi_\Pi}(w_\tau) |\cFtaum] ] \;. 
\$
Let us investigate the inner expectation. Conditional on $\cFtaum$, the iterate $\wtau$ is deterministic, and the distribution of $\ztau | \cFtaum$ is $\Ptaugiventaum$ by definition.
\$
| \E[ F(w_\tau, z_\tau) - {\phi_\Pi}(w_\tau) |\cFtaum]   | 
& = \bigg | \E\bigg[ \int_\cZ F(w_\tau, z) \Ptaugiventaumdz - \int_\cZ F(w_\tau, z)\diff \Pi (z)|\cFtaum\bigg]\bigg |
\\
& \leq  \E\Bigg[ \bigg | \int_\cZ  F(w_\tau, z)\Ptaugiventaumdz - \int_\cZ F(w_\tau, z)\diff \Pi (z)\bigg | |\cFtaum\Bigg]
\\ 
& \leq \bar F \int_\cZ | \diff \Ptaugiventaum  - \diff \Pi (z)| 
\\
& = 2\bar F \cdot \| \Ptaugiventaum - \Pi \|_\TV \;.
\$
where we use boundedness of $F$, i.e., $\sup_w F(w, z) \leq \bar F$ for $\Pi$-almost every $z$,

Next, sum over $\tau = 1,\dots, t$ and move $|\cdot |$ inside the sum and the outer expectation.
\$
    |\E[\I ]| 
    & =  \bigg| \sumtau   \E [   \E[F(w_\tau, z_\tau) - {\phi_\Pi}(w_\tau) | \cFtaum  ]]  \bigg|
    \\
    & \leq  \sumtau  \E \Big[ \big|  \E[F(w_\tau, z_\tau) - {\phi_\Pi}(w_\tau) | \cFtaum  ]  \big|\Big]
    \\
    & \leq 2\bar F \cdot \sumtau \E\big[ \| \Ptaugiventaum - \Pi \|_\TV \big]
    \\
    & \leq 2\bar F \cdot \sup_{z_1,\dots,\zt} \sumtau \| \Ptaugiventaum - \Pi \|_\TV
    \\
    &  = 2\bar F \delta_t \;.
\$

Next consider $|\E[\II ]| $. The analysis goes through without the outer expectation.

Combining we get 
\$\E R_t({w^*_\Pi}) = \E[\I+ \II + \III] \geq \E[\I+ \II] \geq - \big(|\E[\I]| + |\E[\II]|\big) \geq -4\bar F \delta_t \;.\$
\finishproof{\cref{lm:mildcorruption}}
\end{proof}
\subsection{Theorem Statements for Markov Case and Periodic Case} \label{sec:da-ergodic-and-periodic} 
Results for other input types, $\sfC^{\mathrm{E}}(\delta, \iota)$ and $\mathsf{C}^{\mathrm{P}}(q)$, can be obtained by using more complicated regret decompositions and conditioning arguments.  We state the resulting convergence results here. The proofs can be found in \cref{sec:mixing} and \cref{sec:period}.

\begin{theorem}[DA Convergence, Ergodic Case] \label{thm:maintext-da-ergodic}
  For the input distribution $P$ define the $\iota$-step deviation from $\Pi$ for an integer $1\leq \iota \leq t-1$:
  \$
       \epsiota \defeq  \sup_{z_1,\dots, z_t} \sup_{\tau = 1,\dots, t-\iota} \| \Ptaupigiventau - \Pi \|_\TV \;.
  \$
  Then, for all $t\geq 1$ and any $1 \leq \iota \leq t-1$, 
\$
    \E_{\zts \sim P} \big[ \| w_{t+1} - {w^*_\Pi}\|_2^2 \big]
    & 
    \leq 
    \frac{\big(6+\log t\big)G\sq}{\sigma\sq t}   + \frac{2\big(4\bar F\epsiota t + 2G \sq \iota (\log t+1) + 2\iota \bar F\big)}{\sigma t}
    \\
    & =
    \tilde O( \epsiota +  \iota /t ) 
    \;.
\$
Moreover, the rate $ \tilde O( \epsiota +  \iota /t ) $ applies to $\E\big[ \| w_{t+1} - w^*_\gamma\|_2^2 \big]$ and $\E\big[ \| w^*_\gam - w^*_\Pi \|_2^2 \big]$ by the deduction in \cref{sec:deduction-to-hindsight-optima}.
\end{theorem}

\begin{remark}[Comparison with EMD \cite{duchi2012ergodic}]
  Now we specialize \cref{thm:maintext-da-ergodic} to the setting of \cref{rm:markov-input}, and we briefly compare our result with the Ergodic Mirror Descent (EMD) results of \citet{duchi2012ergodic}. EMD considers nonsmooth convex optimization problems of the form $f^* = \min \big\{ f(w)= {\E}_{\Pi}[F(w ; \xi)] \given w\in \cW\big\}$ for a closed convex set $\cW$. Differently from our setting, they do not assume strong convexity in $f$, and do not allow a composite term $\Psi$ which is not linearized. 
  Assume the Markov chain that generates $\zts$ are fast mixing with $\epsilon_t(\iota) 
  \leq M\rho^\iota$ for some $M>0$ and $\rho \in [0,1)$, then the EMD algorithm produces iterates that satisfy the 
  following convergence rate~\footnote{In Eq.\ (3.2) of \cite{duchi2012ergodic}, set 
  $\kappa_1 = M$ and $\kappa_2 = 1/\log(\rho\inv)$ and ignore the parameters $(G,D,\kappa_1)$.}
  \$
  \E \big[ f(w_{t+1}) - f\st \big]  = \tilde O 
  \Bigg( 
    \bigg( 1+\frac{1}{\log(\rho \inv)} \bigg) ^{1/2}
  \cdot 
    \frac{1}{\sqrt{t}}
  \Bigg)
    \;.
  \$
  As in \cref{rm:markov-input}, for the same fast mixing Markov chain, we set $ \iota =  O\Big(\frac{\log t}{\log(\rho\inv )} \Big)$ and $ \epsilon_t(\iota)  = 1/t$ in \cref{thm:maintext-da-ergodic}, we obtain the rate
  \$
  \E \big[ \| w_{t+1} - w\st \|\sq \big] = \tilde O 
  \Bigg(
    \bigg( 1+\frac{1}{\log(\rho \inv)} \bigg) \cdot \frac1t 
  \Bigg)
  \;.
  \$
  which is also the rate for $\E [ f(w_{t+1}) - f\st ]$. Both results characterize the dependence of convergence rate on $\rho$, the mixing parameter of the Markov chain. However,
  our result exploits the strong convexity of the optimization problem and achieves the faster rate $1/t$, while also achieving convergence in iterates rather than only in function values.
  \end{remark}

\begin{theorem}[DA Convergence, Block-Independent Case]\label{thm:maintext-da-periodic}
Fix an integer $K \geq 1$. Let $\{1=\tau_1 < \tau_2 < \dots \tau_{K+1} = t\}$ be an increasing subsequence of $[t]$. Using each two consecutive points, form the interval $I_{k} \defeq [\tauk, \taukp - 1]$. 
Then $\cP \defeq \{ I_k \}_{k=1}^K$ is a partition of 
$[t]$.
Define by $|I_k| \defeq \taukp - \tauk\geq 1$ the length of the interval and $|\cP|_\infty \defeq \max_{k} | I_k|$ the maximum length of the intervals. Associated with the input distribution $P$ and the partition $\cP$ define the block-wise deviation from $\Pi$ by
\$
\delta^{\mathrm{b}} \defeq \frac1t \sumk |I_k| \cdot \bigg\|    \Pi  - \frac{1}{|I_k|} \sumtauinIk  P^\tau \bigg\|_\TV \; .
\$
Assume $\zts$ are block-wise independent according to the partition $\cP$. Then, for all $t\geq 1$, 
\# \label{eq:rateperiodic-in-text}
  \E_{\zts \sim P} \big[ \| w_{t+1} - {w^*_\Pi}\|_2^2 \big] 
  & \leq
  \frac{(6 + \log t) G^2 }{\sigma^2 t} + 
  \frac{ 2\big( 4\bar F \cdot \delta^{\mathrm{b}} t + G\sq |\cP|_\infty\sq( \log t + 1)\big)}
  {\sigma t}  \notag
  \\ 
  & 
  =
  \tilde O ( \delta^{\mathrm{b}} + |\cP|\sq_\infty /t )
  \;.
\#
Moreover, the rate $ \tilde O ( \delta^{\mathrm{b}} + |\cP|\sq_\infty /t )$ applies to $\E\big[ \| w_{t+1} - w^*_\gamma\|_2^2 \big]$ and $\E\big[ \| w^*_\gam - w^*_\Pi \|_2^2 \big]$ by the deduction in \cref{sec:deduction-to-hindsight-optima}.
\end{theorem} 

  Let us briefly we comment on the dependence on $|P|_\infty$.
  Suppose there are in total $K$ blocks, each of equal length $|\cP|_\infty = q$, and blocks are i.i.d.\  We still allow arbitrary dependence within a block. Moreover, we choose $\Pi = \bar P$ in the definition of $\delta^{\mathrm{b}}$. This implies $\delta^{\mathrm{b}}=0$ and then the rate in \cref{thm:maintext-da-periodic} is $q\sq / t$.

  Consider dual averaging {with} the knowledge of the block structure $q$.  Then the rate $1/K = q/t$ can be achieved by executing DA using one randomly chosen data point within a block, throwing away the rest in that same block. Such selection produces $K$ i.i.d.\ samples from $\bar P$. In comparison, the rate in \cref{eq:rateperiodic-in-text} is worse off by a factor of $q$ due to not knowing the block-structure information.

\subsection{Proofs for Markov Case} \label{sec:mixing}

Now we consider data that are not necessarily independent across time. We restrict our attention to ergodic processes, meaning data tend to be independent as they grow apart in time.

Define the $\iota$-step deviation from stationarity
\$
        \epsiota \defeq  \sup_{z_1,\dots, z_t} \sup_{\tau = 1,\dots, t-\iota} \| \Ptaupigiventau - \Pi \|_\TV \;.
\$
An equivalent quantity is the $\epsilon$-mixing time \cite{duchi2012ergodic}
\$
    \iota_{\mix}(\epsilon) \defeq \min \Big\{\iota : 1 \leq \iota \leq t-1, \sup_{z_1,\dots, z_t} \sup_{\tau = 1,\dots, t-\iota} \| \Ptaupigiventau - \Pi \|_\TV \leq \epsilon \Big\} \;.
\$

This means, no matter where and when we start the process, it takes only $\iota$ steps to get $\epsiota$-close to the stationary distribution $\Pi$. 
One could expect the deviation $\epsiota$ decreases as $\iota$ increases. This makes sense because for large $\iota$, the process has run long enough to reach stationarity. 

\begin{theorem}[Mixing Data, Restatement of \cref{thm:maintext-da-ergodic}]  \label{thm:mixing}
It holds for all $t\geq 1$ and any $1 \leq \iota \leq t-1$, 
\#
    \E \big[ \| w_{t+1} - {w^*_\Pi}\|_2^2 \big] = 
        O\Big(
            \frac{\log t}{\sigma^2 t} 
        + \frac{\iota \log t}{\sigma  t}
        + {\epsiota}/{\sigma }
        \Big) \,,
\#
    where $O(\cdot)$ hides dependence on constants, $G$ and $\barF$. Here there is a trade-off in $\iota$ in the last two terms.
\end{theorem}

\emph{Proof of \cref{thm:mixing}}
We use the proof in \cite{duchi2012ergodic}; see Eq.\ (6.2) in the paper. Decompose $R_t(\wst)$ as follows.
\$
R_t(\wst) 
= & \sumtautmi \Big( \big(F(\wtau,\ztaupi) - F(\wst,\ztaupi)\big) - \big({\phi_\Pi}(\wtau) - {\phi_\Pi}(\wst) \big)\Big) \tag{A}
\\
& + \sumtautmi \big(F(\wtaupi,\ztaupi) - F(\wtau,\ztaupi)\big) \tag{B}
\\
& + \sumtautmi \big({\phi_\Pi}(\wtau) - {\phi_\Pi}(\wst) \big) \tag{C}
\\
& + \sumtaui \big(F(\wtau,\ztau) - F(\wst,\ztau)\big)  \;. \tag{D}
\$

By optimality of $\wst$ we have $ \rC \geq 0$. By boundedness of $F$ we get $|\rD| \leq 2\iota \bar F$. Remains to handle A and B. We will show
\[ \text{"}\,\, \rA \leq \epsiota t, \quad \rB\leq \iota \,\,\text{"} \,.\]

\underline{Bounding A.} The key is $\ztaupi$ is almost independent of $\cFtaum$ if $\iota$ is moderately large. For each $\tau$,
\$
& |\E[F(\wtau,\ztaupi) - {\phi_\Pi}(\wtau)]| 
\\
& = |\E[ \E[ F(\wtau,\ztaupi) -  {\phi_\Pi}(\wtau) | \cFtaum]]|
\\
& =  \Bigg | \E\Bigg[  \E\Bigg[  \int_\cZ  F(w_\tau, z) \Ptaupigiventaumdz  - \int_\cZ F(w_\tau, z) \Pi (\diff z) |\cFtaum\Bigg] \Bigg] \Bigg | \tag{Key}
\\
& \leq   \E \Bigg[ \E\Bigg[  \Big |\int_\cZ  F(w_\tau, z) \Ptaupigiventaumdz  -  \int_\cZ F(w_\tau, z) \Pi (\diff z)\Big | |\cFtaum\Bigg] \Bigg]
\\ 
& \leq 2\bar F \cdot \E \big[ \| \Ptaupigiventaum - \Pi \|_\TV \big]
\\
& \leq 2\bar F \cdot \sup_{z_1,\dots, z_{\tau - 1}} \| \Ptaupigiventaum - \Pi \|_\TV \leq 2\bar F \epsiota \;.
\$
Analysis for $|\E[F(\wst,\ztaupi) - {\phi_\Pi}(\wst)]| $ is almost identical. Next sum over $\tau = 1,\dots, t-\iota$.
\$
| \E[ \rA ]| \leq 4\bar F \epsiota \cdot t \;.
\$

\underline{Bounding B.} The change in $F$ by $\iota$ steps of updates, starting from $\wtau$, is controlled by $c\cdot \iota G \cdot \frac1\tau$ where $1/\tau$ acting like a stepsize.

\begin{lemma} \label{lm:onestep}
Let $\Pi_{\Psi, \cW} (g) \defeq\argmin_{w\in \cW} \{ \langle g, w \rangle + \Psi(w) \}$. If $\Psi$ is $\sigma$-strongly convex, then 
    \$
    \| \Pi_{\Psi, \cW} (g) - \Pi_{\Psi, \cW} (g')\| \leq (1/\sigma) \|g - g'\|_*\;.
    \$
\end{lemma}
\begin{proof}
    See \cite[Lemma 6.1.2]{nesterov2003introductory}.
\end{proof}
Noting $\wtaupone = \Pi_{\Psi, \cW}(\gbartau)$ and $\wtau = \Pi_{\Psi, \cW}(\gbartaum)$, \cref{lm:onestep} gives
\# \label{eq:onestepdiff}
\| w_{\tau + 1} - \wtau \| \leq \| \gbartau - \gbartaum \|_* /\sigma  = \| \gbartaum - \gtau\|_* / (\tau \sig) \leq 2G / (\tau \sig) \;.
\#
It holds $\Pi$-a.s.\ that for each $\tau$, the map $w \mapsto F(w,\ztaupi)$ is Lipschitz with parameter $G$.
\$
|\E[F(\wtaupi,\ztaupi) - F(\wtau,\ztaupi)]|
& \leq G \cdot \E \big[\| \wtaupi - \wtau\|\big]
\\
& \leq G \cdot \sum_{t' = \tau}^{\tau + \iota - 1} \E\big[\|w_{t'+1} - w_{t'} \|\big]
\\
& \leq G \cdot \sum_{t' = \tau}^{\tau + \iota - 1} 2G / (\sigma t')
\\
& \leq  G \cdot \sum_{t' = \tau}^{\tau + \iota - 1} 2G / (\sigma \tau) = 2G\sq \iota/ \tau \;.
\$
Summing over $\tau = 1,\dots, t-\iota$, we get
\$
| \E[ \rB] | \leq 2G\sq \iota (\log t + 1)\;.
\$

Putting together, 
\$
\E[R_t(\wst)] 
&= \E[ \rA+\rB+ \rC+ \rD]
\\
& \geq \E[ \rA+ \rB+ \rD] 
\\
& \geq - \big(|\E[\rA]| +|\E[\rB]| +|\E[ \rD]| \big)
\\
& \geq - (4\bar F\epsiota t + 2G \sq \iota (\log t + 1) + 2\iota \bar F)\;,
\$
and 
\# \label{eq:concrete-markov bound}
\E[\left\|w_{t+1}-w\right\|^{2} ] 
& \leq \frac{1}{\sigma t} \big(\E[\Delta_t]  - \E[R_t(\wst)]\big) \notag
\\
& 
\leq \frac{\big(6+\log t\big)G\sq}{\sigma\sq t}   + \frac{2\big(4\bar F\epsiota t + 2G \sq \iota (\log t+1) + 2\iota \bar F\big)}{\sigma t}\;.
\#
We complete the proof of \cref{thm:mixing}.

\subsection{Proofs for Periodic Case} \label{sec:period}

    Assume $\zts$ are block-wise independent according to the partition $\cP$. Given the partition $\cP$, define 
    \$
        \delta^\block_t \defeq  \sumk |I_k| \cdot \Bigg\|    \Pi  - \frac{1}{|I_k|} \sumtauinIk  P^\tau \Bigg\|_\TV \; .
    \$
    Note we compute the deviation in a block-wise manner. Note $ \delta^\block_t = t \cdot \delta^\mathrm{b}$ with $\delta^\mathrm{b}$ defined in \cref{thm:maintext-da-periodic}.

\begin{theorem}[Block-wise Independent Data, Restatement of \cref{thm:maintext-da-periodic}] \label{thm:block}
    It holds 
    \#
        \E \big[ \| w_{t+1} - {w^*_\Pi}\|_2^2 \big] = 
            O\bigg(
                \frac{\log t}{\sigma^2 t} 
            + \frac{ | \cP|_\infty\sq \log t}{\sigma  t}
            + \frac{\delta^\block_t }{\sigma t}
            \bigg) \;,
    \#
    where $O(\cdot)$ hides dependence on constants, $G$ and $\barF$.
\end{theorem}

Generally, compared with $\delta_t$ defined in \cref{def:deltat}, our new notion of deviation can be much smaller for block-structured data. This is especially true when each block of data, as a whole, forms a good estimate of $\Pi$, but each data point in the block deviates from $\Pi$ by a constant amount. The periodic case in \cref{rm:gainfromblock} examplifies this.

\begin{remark}[Extreme 1: Recover Independent Case]
    Setting $|\cP|_\infty = 1$ and $\delta_t^\block = \delta_t$ in \cref{def:deltat} we recover the usual rate under independence assumption (\cref{thm:corrupted}).
\end{remark}

\begin{remark}[Extreme 2: Fail to Recover Arbitrary Distribution Case]
If we allow arbitrary dependence in the whole sequence $\gamma = \zts$, then we can only set $|\cP|_\infty = t$ and the bound is useless.
\end{remark}

\begin{remark}[The Gain from Block Structure] \label{rm:gainfromblock}

    Although \cref{thm:corrupted} applies to block-structure data, we obtain significant improvement in \cref{thm:block}.

    Consider the periodic case where each block is of length $q$ and blocks are i.i.d.\ At the start of block $I_k$, we draw a sample from $\Pi$, i.e., $z_{t_k} \sim \Pi$, and then let rest of the $z_\tau$'s in that block equal $z_{t_k}$.
    In this case $\delta_t^\block=0$ because the marginal of every $\ztau$ is exactly $\Pi$. Then the bound in \cref{thm:block} becomes
    \#  \label{eq:rateperiodic}
    \frac{q\sq \log t}{t} \;,
    \#
    which converges to zero at the rate $q\sq / t$. However, the bound in \cref{thm:corrupted} fails to converge. To see this let us estimate $\delta_t$ in \cref{def:deltat}. Consider some $\tau$ in the interval $I_k$. If $ \tau \neq t_k$, then the conditional distribution $P^\tau(\cdot \given z_{1:\tau-1})$ is a point mass on $z_{t_k}$. If $\tau = t_k$ then $P^\tau(\cdot \given z_{1:\tau-1}) = \Pi$. Let $c=\sup_z \| \delta_z - \Pi \|_\TV$. The quantity $c$ is positive unless $\Pi$ is a point mass. Then
    \$
      \delta_t = \sup_{z_1,\dots,\zt} \sumk \sumtauinIk \| P^\tau(\cdot \given z_{1:\tau-1}) - \Pi \|_\TV = \sup_{z_1,\dots,\zt} \sumk (q-1) \| \delta_{z_{t_k}} - \Pi \|_\TV = K(q-1) c\; ,
    \$ 
    and then \cref{thm:corrupted} becomes
    \$
    \frac{\log t}{t} + c \to c\;.
    \$
\end{remark}

\emph{Proof of \cref{thm:block}}

We decompose the regret by blocks.
\# \label{eq:regret_periodic}
    R_t({w^*_\Pi}) =& 
  \sumk\sumtauinIk \big(F(w_\tau, z_\tau) - {\phi_\Pi}(w_\tauk) \big) +
  \sumtau \big({\phi_\Pi}({w^*_\Pi}) - F({w^*_\Pi}, z_\tau) \big) 
  \\ & +
  \sumk\sumtauinIk \big({\phi_\Pi}(w_\tauk) - {\phi_\Pi}({w^*_\Pi})\big) \;.
\#
Rewrite the first sum by adding and then subtracting the term $\sumk\sumtauinIk F(w_\tauk, \ztau)$.
\$
& \sumk\sumtauinIk \big(F(w_\tau, z_\tau) - {\phi_\Pi}(w_\tauk)\big) 
\\= & \sumk \Bigg( \underbracket{ \sumtauinIk \big(F(w_\tau, z_\tau) - F(w_\tauk, z_\tau) \big)}_{\defeq B_k} \Bigg) \tag{I}
  \\
 &+  \sumk  \Bigg( \underbracket{ \sumtauinIk \big(F(w_\tauk, z_\tau) - {\phi_\Pi}(w_\tauk) \big) }_{\defeq A_k}\Bigg)  \;.\tag{II}
\$
\underline{Bounding $A_k$.}
Use a conditioning argument. The key is, conditional on $\cF_{\tauk - 1}$, the iterate $w_\tauk$ is deterministic and the distribution of $\ztau$ is $P^\tau$ due to block-wise independence.
\$
|\E[A_k]| 
    & =  \Bigg| \sumtauinIk   \E [   \E[F(w_\tauk, z_\tau) - {\phi_\Pi}(w_\tauk) | \cF_{\tauk - 1} ]]  \Bigg|
    \\
    & \leq   \E \Bigg[ \bigg|  \sumtauinIk \E[F(w_\tauk, z_\tau) - {\phi_\Pi}(w_\tauk) | \cF_{\tauk - 1}  ]  \bigg|\Bigg]
    \\  
    & =   \E \Bigg[ \bigg|  \sumtauinIk \int_\cZ F(w_\tauk, z) P^\tau(\diff z) - \int_\cZ F(w_\tauk, z)\Pi(\diff z)    \bigg|\Bigg]
    \\
    & \leq 2\bar F \cdot  \bigg\| \sumtauinIk \big(P^\tau - \Pi\big) \bigg\|_\TV  \;.
\$
Then sum over $k=1,\dots, K$, and we have
\$
\big| \E[\I] \big|  \leq \sumk \big|\E[A_k ]\big| \leq 2\bar F \cdot 
 \sumk \bigg\| \sumtauinIk (P^\tau - \Pi) \bigg\|_\TV
 \leq 2\bar F \delta_t^\block \;.
\$
\underline{Bounding $B_k$.} Using \cref{lm:onestep} and \cref{eq:onestepdiff}, we have
\$
| \E[B_k] | 
& 
\leq G \cdot \sumtauinIk \E \big[\| \wtau - w_{\tauk} \|\big] 
\\
& 
\leq G \cdot \sumtauinIk G (\tau - \tauk) / \tauk 
\\
& 
\leq G \cdot \sumtauinIk G (\taukp - \tauk) / \tauk
\\
& = G\sq (\taukp - \tauk)\sq /\tauk \leq G\sq |\cP|_\infty\sq /\tauk \;.
\$
Then sum over $k=1,\dots, K$, and we have
\$
\big| \E[\II]\big | \leq \sumk \big|\E[B_k ]\big| 
= G\sq |\cP|_\infty\sq  \cdot \sumk \frac{1}{\tauk}
\leq  G\sq |\cP|_\infty\sq \cdot \sumtau \frac1\tau
\leq  G\sq |\cP|_\infty\sq ( \log t + 1)\,.
\$

It can be shown the second sum in the regret decomposition (\cref{eq:regret_periodic}) is upper bounded by $2\bar F \delta_t^\block$. The third sum is $\geq 0$. Using \cref{fact:linxiao-intext} we get
\# \label{eq:concrete-block-bound}
\E[\left\|w_{t+1}-w\right\|^{2} ] 
\leq \frac{1}{\sigma t} \big(\E[\Delta_t]  - \E[R_t(\wst)]\big) 
\leq \frac{(6 + \log t) G^2 }{\sigma^2 t}   + 
\frac{ 2\big( 4\bar F \delta_t^\block + G\sq |\cP|_\infty\sq( \log t + 1)\big)}
{\sigma t} \;.
\#
We complete the proof of \cref{thm:block}.

\section{Proofs for PACE}
\label{sec:proof_for_pace}

\subsection{Proof of \cref{thm:maintext-beta-conv-indep} and \cref{thm:maintext-beta-conv-ergodic-periodc}}
\label{appsub:thm:beta-various-cases-intex}

We show convergence of $\beta$ under different input assumption. Recall $\beta^{t}$ is the pacing multiplier generated by PACE, and $\betast$ is the solution to the optimization problem \cref{eq:dual-beta-infinite}. The vector $\gamma$ is the sequence of items.

\begin{theorem}[Restatement of \cref{thm:maintext-beta-conv-indep} and \cref{thm:maintext-beta-conv-ergodic-periodc}] \label{thm:beta-various-cases}
    For the independent case, i.e., $\gamma \sim Q$ and $Q\in  \mathsf{C}^\mathrm{ID}(\delta) $, it holds for $t \geq 1$
    \# \label{eq:beta-conv-1}
    \E[\|\beta^{t+1} - \betast \|\sq ] \leq   \frac{(6 + \log t) G^2 }{\sigma^2 t} + \frac{8\bar F }{\sigma }\delta 
    \;.
    \#
    and for $t \geq 3$, 
\# \label{eq:beta-conv-2}
    \frac1t \sumtau \E[\|\betatau - \betast \|\sq ] \leq  \frac{G\sq}{\sigma\sq} \left(6(1+\log t)+\frac{(\log t)^{2}}{2}\right) + (8\barF /\sigma ) \cdot \delta 
    \;.
\#
For the erogdic case, i.e., $\gamma \sim Q$ and $Q\in \sfC^{\mathrm{E}}(\delta, \iota)$, it holds for $t \geq 1$, 
\# \label{eq:beta-conv-3}
\E[\|\beta^{t+1} - \betast \|\sq ] \leq   \frac{\big(6+\log t\big)G\sq}{\sigma\sq t}   + \frac{2\big(4\bar F \delta t + 2G \sq \iota (\log t+1) + 2\iota \bar F\big)}{\sigma t} = \tilde O\bigg( \delta + \frac{\iota}{t}\bigg).
\#
and for $t \geq 3$, 
\# \label{eq:beta-conv-4}
\frac1t \sumtau \E[\|\betatau - \betast \|\sq ]=\tilde O\bigg( \delta + \frac{\iota}{t}\bigg)
\;.
\#
For the periodic case, i.e., $\gamma \sim Q$ and $Q\in \mathsf{C}^{\mathrm{P}}(q) $, it holds for $t \geq 1$
\# \label{eq:beta-conv-5}
\E[\|\beta^{t+1} - \betast \|\sq ] \leq 
\frac{(6 + \log t) G^2 }{\sigma^2 t}   + 
\frac{ 2  G\sq q \sq( \log t + 1)}
{\sigma t}  = \tilde O ( q\sq / t  )\;.
\#
and for $t \geq 3$, 
\# \label{eq:beta-conv-6}
\frac1t \sumtau \E[\|\betatau - \betast \|\sq ]=\tilde O(q\sq/t)
\;.
\#
\end{theorem}
\begin{proof}[Proof of \cref{thm:beta-various-cases}]
Set $\sigma = 1/n$ and $G=\barF = \vinfty$. 
\cref{eq:beta-conv-1} follows by \cref{thm:corrupted} and specifically \cref{eq:concrete-indep-bound}. The next inequality \cref{eq:beta-conv-2} follows by \cite[Corollary 4]{xiao2010dual}: for $t\geq 3$,
\$
\frac{1}{t} \sum_{\tau=1}^{t} \frac{(6+\log \tau) G^{2}}{\tau \sigma^{2}} \leq \frac{1}{t}\left(6(1+\log t)+\frac{(\log t)^{2}}{2}\right) \frac{G^{2}}{\sigma^{2}} 
\;.
\$
\cref{eq:beta-conv-3} follows by \cref{thm:mixing} and specifically \cref{eq:concrete-markov bound}. Following \cite[Corollary 4]{xiao2010dual}, we have for $t \geq 3$, 
\$
& \sumtau \E[\|\betatau - \betast \|\sq ]  
\\
& \leq \frac{G\sq}{\sigma\sq} \left(6(1+\log t)+\frac{(\log t)^{2}}{2}\right) + 8\barF  \cdot \delta t + \frac{4G\sq}{\sigma}  \left(6(1+\log t)+\frac{(\log t)^{2}}{2}\right) + \frac{4 \barF}{\sigma}(1+\log t) \cdot \iota
\\
& = \tilde O( \delta t + \iota)
\;,
\$
and thus \cref{eq:beta-conv-4} holds.

The inequality \cref{eq:beta-conv-5} follows from \cref{thm:block} and specifically \cref{eq:concrete-block-bound}. For the inequality \cref{eq:beta-conv-6}, apply the same strategy: for $t \geq 3$, 
\$
& \sumtau \E[\|\betatau - \betast \|\sq ]  
\\
& \leq \frac{G\sq}{\sigma\sq} \left(6(1+\log t)+\frac{(\log t)^{2}}{2}\right) + 
\frac{2G\sq }{\sigma\sq} \left(6(1+\log t)+\frac{(\log t)^{2}}{2}\right) \cdot q\sq
= \tilde O(q\sq)
\;.
\$
\end{proof}

\subsection{Proof of \cref{lm:dual-to-expendicture-util} and \cref{lm:dual-to-regret}}
\emph{Proof of \cref{lm:dual-to-expendicture-util}}

\cref{lm:dual-to-expendicture-util} follows from \cite[Theorem 3 and 4]{gao2021online}.

\emph{Proof of \cref{lm:dual-to-regret}}

Define the hindsight \emph{average} equilibrium utility $u^\gam_i \defeq (1/t) \cdot U^\gam _i$. Although results in  \cite{gao2021online} were stated for i.i.d.\ case, the proof in fact goes through for nonstationary input distributions.

\underline{Bounding $\E[\| \bar u^t - u^\gam \|\sq]$.}
For the first inequality we use the proof of \cite[Theorem 6]{gao2021online}.
Follow that paper, we define $r^t_i = \max\{0, \bar u^t_i - u^\gam_i \}$.
In Theorem 6 the authors show
\$
 \E [ (r^t_i)^{2}] \leq   C_{r,1} \cdot \E[ \| \betatp - \betast \|\sq ] + C_{r,2} \cdot \bigg(\frac1t \sumtau \E[\| \betatau - \betast \|\sq]\bigg)
 \;.
\$
In particular, the constant $C_{r,1}$ comes from the constant $C$ in \cite[Theorem 4]{gao2021online} and $C_{r,2}$ comes from \cite[Equation (11)]{gao2021online}

\underline{Bounding $\Reg_{i,t}$.}
Note $\Reg_{i,t} = t \cdot (u^{\gam}_i - \bar{u}^{t}_i) \leq t \cdot r^t_i$. Then we use Cauchy-Schwarz.
\$
\E \big[ \Reg_{i,t}\big] 
\leq t \E[ r^t_i ]
\leq  t \sqrt{\E[(r^t_i)\sq ]}
\;.
\$

\underline{Bounding $\Envy_{i,t}$.}
For the second inequality we use the proof of Theorem~6 in the same paper. Following that paper, we define \$\rho^t_i = (n/t) \cdot \max_{k\in [n]} \Big\{ \langle v_i(\gamma), x_k\rangle - \langle v_i(\gamma), x_i\rangle \Big\} 
\;.
\$ 
During the course of proving Theorem~6, the authors show
\$
\E[(\rho^t_i)\sq ] \leq n\sq \bigg(C_{e,1}\cdot \E[ \| \betatp - \betast\|\sq ] + C_{e,2} \cdot \Big(\frac1t \sumtau \E[\| \betatau - \betast \|\sq]\Big)\bigg)
\;.
\$
In particular, the constant $C_{e,1}$ comes from \cite[Theorem 4, Equations (5) and (13)]{gao2021online} and 
$C_{e,2}$ comes from  \cite[Equations (5), (13) and (15)]{gao2021online}

Then using Cauchy-Schwarz,
\$
\E \big[\Envy_{i,t}(\gamma)\big] = \E \big[ (t/n  )\cdot \rho^t_i  \big] \leq (t/n) \sqrt{\E\big[ (\rho^t_i)\sq  \big]}
\;.
\$
\finishproof{\cref{lm:dual-to-regret}}

\section{Experiments} \label{sec:experiments}

We conduct experiments on a market (a matrix of buyers' valuations on items) generated from the MovieLens dataset \citep{harper2016movielens} with $n=100$ buyers and $m=300$ items. The process of turning the MovieLens dataset into the market instance is described in \cite{kroer2021computing}. Here, we briefly describe the experiment settings. For more details on the experiment settings as well as all code and data to replicate the results, please refer to the Supplementary Material.

We generate item arrivals from the following data input models:
\begin{itemize}
    \item i.i.d.: Every item $\theta^t \in [m]$ is sampled independently from a fixed distribution $s^0\in {\Delta}_m$ (an $m$-dimensional probability vector). 
    \item Mild corruption: $\theta^t \sim s^t$, where $s^t\in {\Delta}_m$ is a distribution such that $\|s^t - s^0\|_1 = \Theta(1/t)$ for all $t$. Here, $s^t$ is generated by randomly perturbing each coordinate of $s^0$ followed by normalization.
    \item Markov: $(\theta^t)_{t\geq 1}$ is sampled from an irreducible Markov chain starting from an initial distribution $s^0$. It is a special case of ergodic input. 
    Here, the Markov chain is given by a $m\times m$ transition matrix (each row sums to $1$), which we generated randomly (and row-wise normalized).
    In this case, the ``reference'' item arrival distribution is the stationary distribution of this Markov chain which is in general different from the initial distribution.
    \item Periodic: The period length is $\ell=100$. Let $(s^k)_{k\in [\ell]}$ be a set of distributions (probability vectors). Here, each $s^k$ is sampled randomly and normalized. The item arrivals of each period is generated by sampling from each $s^k$ followed by a random permutation over the $\ell$ sampled items. 
\end{itemize}
For each (fixed) data input model, we generate $10$ sample paths of item arrivals and run PACE for $T=200n=20000$ time steps on each sample path. Then, we measure the convergence of the pacing multipliers and time-averaged cumulative utilities to their hindsight equilibrium values. 
More specifically, we record the following relative differences: $\max_i \{ {|\beta^t_i-\beta^{\rm HS}_i|}/{\beta^{\rm HS}_i}\}$ and $\max_i \{ {|\bar{u}^t_i-u^{\rm HS}_i|}/{u^{\rm HS}_i}\}$, where ${\rm HS}$ denote the hindsight equilibrium values of the ``sample-path'' market determined by the realized item arrivals. 
Equivalently, $u^{\rm HS}$ and $\beta^{\rm HS}$ are optimal solutions of the hindsight convex programs \eqref{eq:EGprogram} and \eqref{eq:dual-beta-finite}, respectively. 
We also measure the performance of a proportional-share baseline solution that divides each arriving item among all buyers proportionally w.r.t.\ their budgets: for an arrived item $\theta^t$, each buyer $i$ gets $B_i$ amount of it and receives utility $B_i v_i(\theta^t)$ (in this paper, the buyers' budgets are $B_i = 1/n$ for all $i$).
We compute the means and standard errors of the error measures across the $10$ sample paths and plot them in Figure~\ref{fig:movielens-small-plots-all-input-models}. 

As can be seen, for all data input models, the pacing multipliers and buyers' time-averaged utilities converge to their respective hindsight values and quickly outperform the baseline proportional-share solution. Similar convergence behavior can also be observed when the error metrics are w.r.t.\ to the true equilibrium values $\beta^*$, $u^*$ instead of the hindsight values.

\begin{figure}
    \centering
    \includegraphics[width=0.48\linewidth]{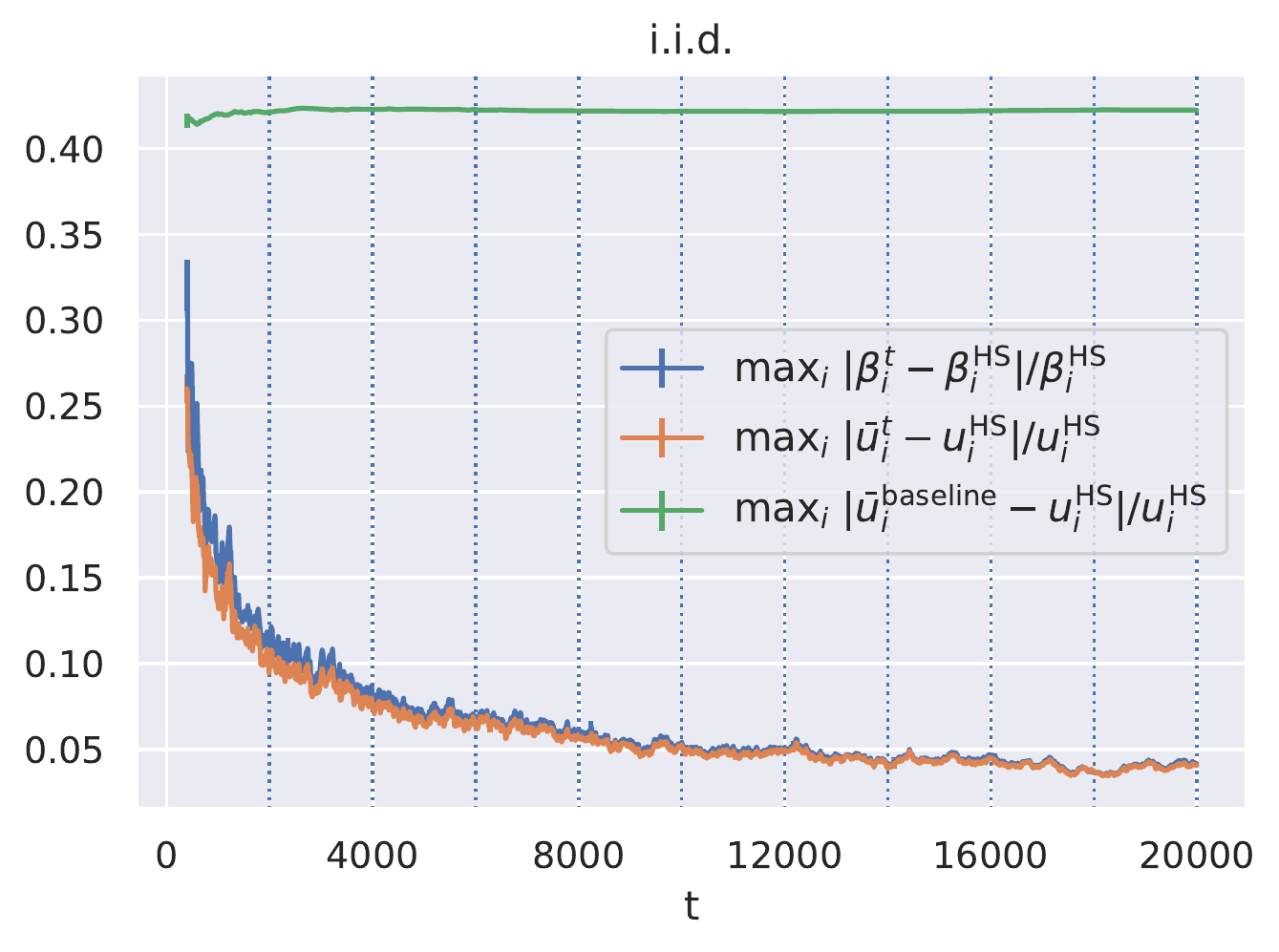}
    \includegraphics[width=0.48\linewidth]{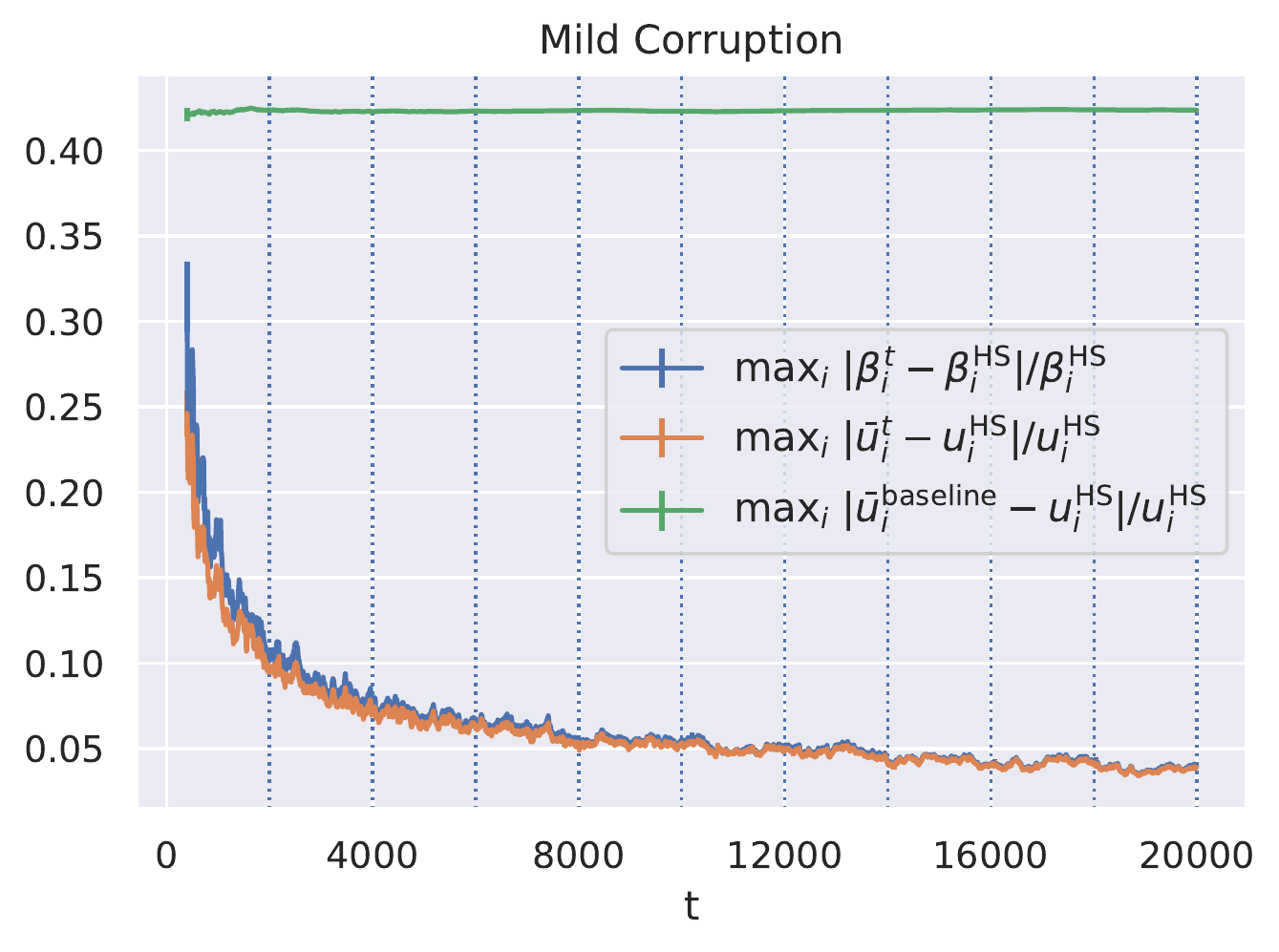} \\
    \includegraphics[width=0.48\linewidth]{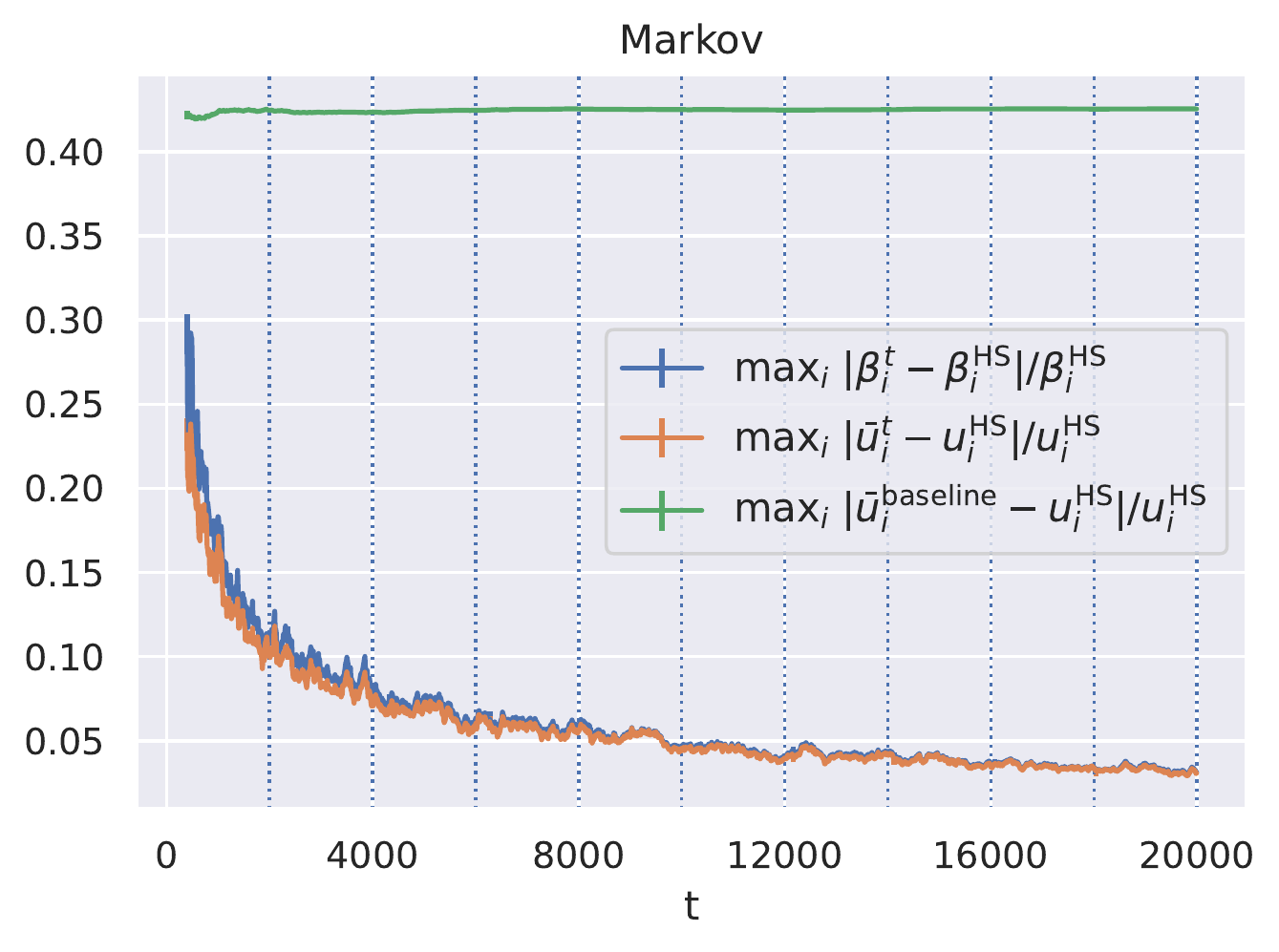}
    \includegraphics[width=0.48\linewidth]{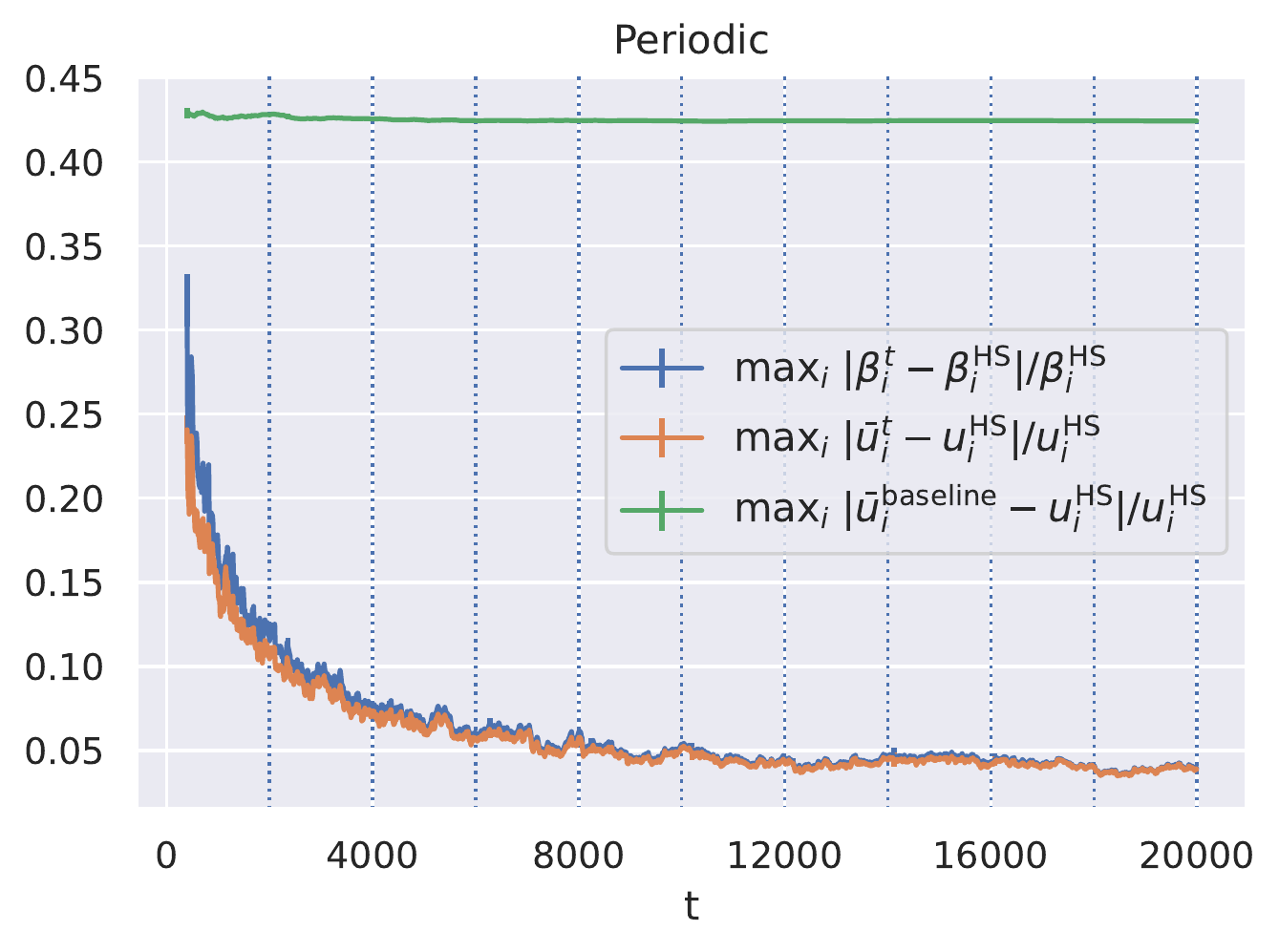}
    \caption{
        Performance of PACE for item arrivals under different data input models. 
        All error measures are averaged across $10$ repeated experiments. 
        The mean and standard errors of the error measures are plotted, where the standard error bars are too small and hence invisible. 
        Here, $\bar{u}^{\rm baseline}_i$ are the buyers' time-averaged utilities under a ``proportional-share'' baseline solution.
    }
    \label{fig:movielens-small-plots-all-input-models}
\end{figure}

We further conduct the following experiment to demonstrate the effect of nonstationarity on pacing multiplier convergence. 
More specifically, we generate $10$ sample paths from each of the following item arrival settings: i.i.d. distributions ($\theta_t \sim_{\rm i.i.d.} s$), perturbed distributions (small $\delta$), perturbed distributions (large $\delta$), where $\delta$ is the overall difference between the sequences of i.i.d. distributions and perturbed distributions, as in \eqref{eq:def-CID(delta)}.
We then run PACE on each sample path to obtain $\beta^t_i$ and $\bar{u}^t_i$ values.
For each setting and each time step, we plot mean values and standard error bars of the relative error metrics, where $\beta^*_i$ and $u^*_i$ are the equilibrium pacing multipliers (utility prices) and utilities of the market with supplies being the distribution $s$.
As can be seen, for perturbed distribution settings, PACE is able to bring $\beta^t$  and $\bar{u}^t_i$ close to their equilibrium values, while the convergence degrade as the perturbation amount $\delta$ increases. 
Recall that, in terms of cumulative utility, the proportional-share baseline solution gives a relative error of around $0.45$.

\begin{figure}
    \centering
    \includegraphics[width=0.48\linewidth]{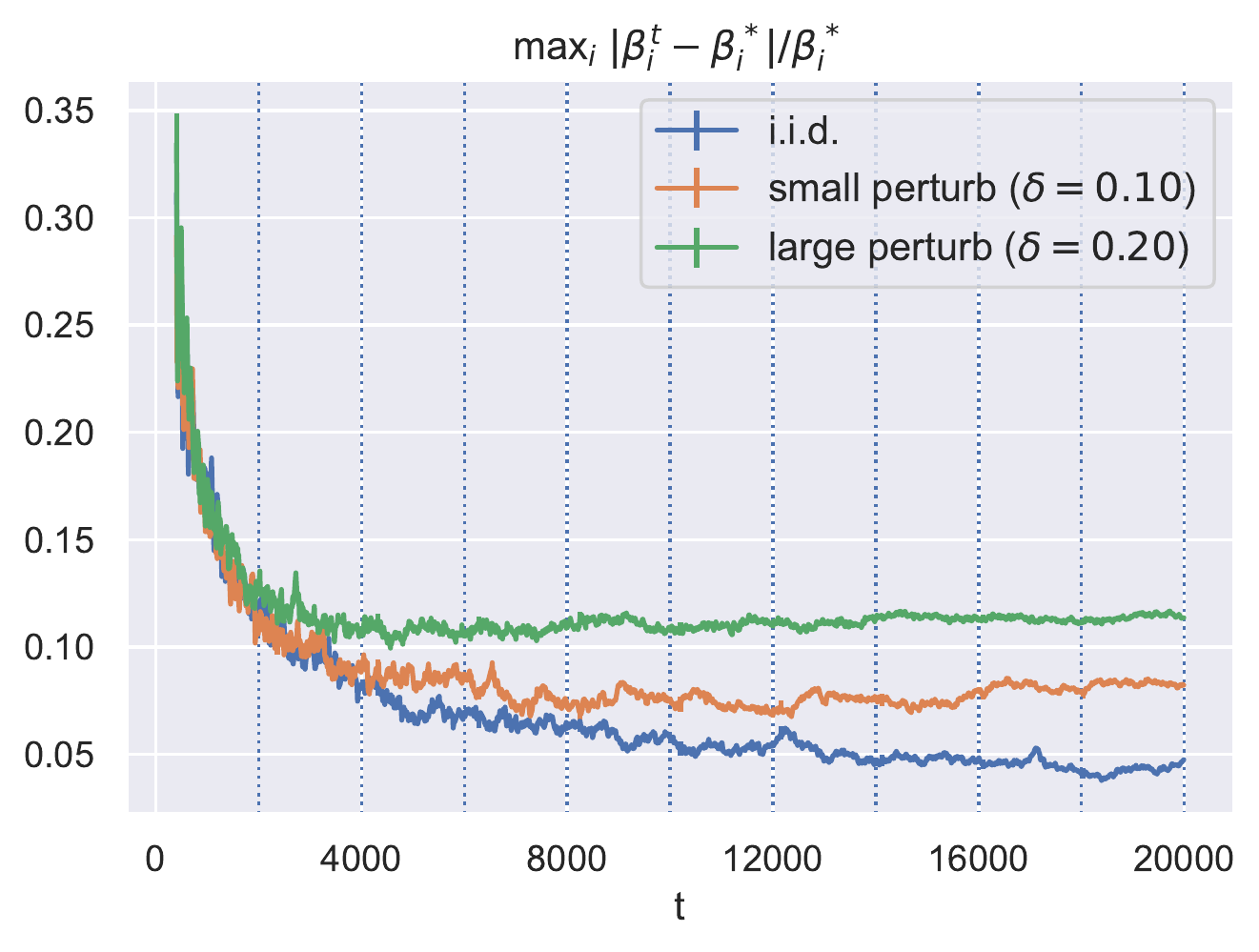}
    \includegraphics[width=0.48\linewidth]{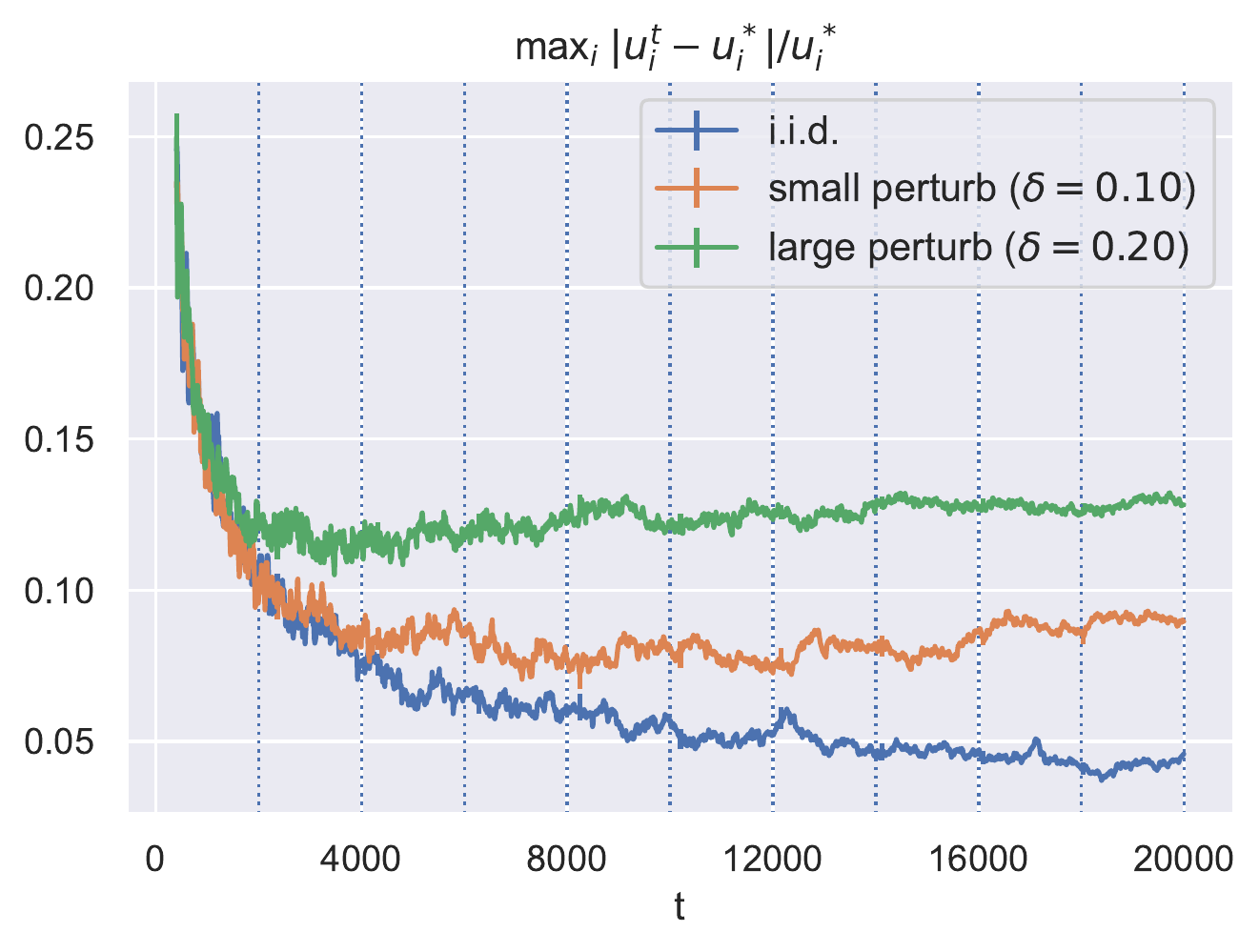}
    \caption{Convergence of pacing multipliers $\beta^t$ and cumulative utilities $\bar{u}^t_i$ under i.i.d. and perturbed distributions}
    \label{fig:compare-iid-perturb}
\end{figure}

\end{document}